\documentclass[journal]{IEEEtran}

\usepackage{amsmath,amssymb,amsfonts,amsthm}
\usepackage{graphicx,xcolor}
\usepackage{booktabs}
\usepackage{cite}
\usepackage{mathtools}
\usepackage{url}
\usepackage{titlesec}
\titlespacing{\section}{0pt}{1ex plus 0.2ex minus 0.2ex}{0.5ex plus 0.1ex}
\titlespacing{\subsection}{0pt}{0.8ex plus 0.2ex minus 0.2ex}{0.3ex plus 0.1ex}
\newtheorem{proposition}{Proposition}

\newtheorem{remark}{Remark}

\newcommand{\CRB}{\mathrm{CRB}}
\newcommand{\defn}{\mathrel{\mathop:}=}

\begin{document}

\title{Detection-Aware Rate--CRB Characterization for Resource Allocation in ISAC Systems}

\author{Sana~Mazahir,~\IEEEmembership{Member,~IEEE,}
and Hina~Tabassum,~\IEEEmembership{Senior Member,~IEEE}
\thanks{The authors are with the Department of Electrical Engineering
and Computer Science, York University, Canada (e-mail:
sana2025@yorku.ca; hinat@yorku.ca).}}

\maketitle

\begin{abstract}
Integrated sensing and communication (ISAC) performance analysis commonly employs the Cram\'er--Rao bound (CRB) as the sensing metric. The CRB lower-bounds estimation-error variance under a target-present model but does not characterize the reliability of the energy detection that receivers often use to validate returns against noise and clutter. Considering the range CRB in a single-user, single-target joint power--bandwidth allocation problem, we establish the conditions under which a rate-optimal allocation at a given CRB requirement fails to achieve a given detection probability. We then develop detection-constrained rate--CRB characterizations at two levels: 1) a detection-aware (DA) formulation that treats the CRB and detection probability as parallel requirements, and 2) a detect--then--estimate (DTE) formulation using a \textit{detection-gated CRB} defined on frames that cross an energy threshold. We analyze the resulting changes in the optimal rate and feasible power--bandwidth allocations, and for DTE further show how detection probability and conditional estimation accuracy jointly determine downstream tracking performance. Multi-user, multi-target resource allocation is then formulated under CRB-only, DA, and DTE characterizations. The results show that CRB-only allocation can conceal inadequate detection reliability, and that accounting for detection and receiver architecture materially changes the power--bandwidth allocation and the communication--sensing tradeoff.
\end{abstract}


\section{Introduction}
Integrated sensing and communication (ISAC) is expected to be a native capability of next-generation wireless networks, enabling communication and sensing to share hardware, waveforms, and time, frequency, and spatial resources. This shared-resource operation has motivated extensive ISAC beamforming and
resource-allocation research~\cite{lu2023ten,liu2022joint}, commonly characterized through paired communication and sensing metrics. To date, the pairing of achievable rate with the Cram\'er--Rao bound (CRB) is the most widely adopted characterization~\cite{liu2021optimaltx,dong2022sensing,liu2022joint,
hua2024crbrate,ren2024fundamental,xiong2023fundamental,
jointAngleDelay2023,li2024uplink,song2026crb,mazahir2025dfrc}.

The CRB lower-bounds the estimation error variance under a
target-present assumption \cite{richards2014radar}. This means that once an initial decision is made that a target is present in a particular delay-Doppler-angle cell, the estimator is assumed to deliver a parameter value on every subsequent sensing frame, including those with weaker target returns. In practice, however, sensing receivers couple parameter estimation with an energy detection mechanism to ensure that measurements are only accepted from validated target echoes. This coupling is realized in different ways: detect-then-estimate (DTE) invokes estimation only after validating the current-frame echo using an energy detector~\cite{richards2014radar,chaumette2005influence,
chaumette2007conditional}; Bayesian joint detection and
estimation couples the two stages through a decision-dependent cost~\cite{moustakides2012joint,tajer2010optimal}; a generalized likelihood ratio test forms the detection statistic by substituting maximum-likelihood estimates of the unknown parameters into the likelihood ratio~\cite{vantrees1971detection3,gaudio2020joint}; and track-before-detect avoids a framewise detection decision by integrating measurements along candidate target trajectories over multiple frames before declaring
a track~\cite{tonissen1996}. Tracking performance is likewise affected \cite{sinopoli2004kalman}, since the tracker updates on successful measurements rather than on the assumption that an estimate is always available.

If an allocation under a prescribed CRB yields an insufficient
signal-to-noise ratio (SNR) for reliable detection, the
CRB can overstate the delivered sensing performance. This
is also consistent with the well-known fact that the CRB
is tight only at a sufficiently high SNR \cite{otfs_isac_zzb_tvt}.
Characterizing the hierarchical detection--estimation--tracking operation then requires an explicit detection requirement and, for a gated receiver model, conditioning the estimation metric on detection. This is especially important when the shared resources are scarce, as the allocation is then more likely to push the sensing performance into a detection-limited regime.

\subsection{Related Work and Research Gap}
Existing ISAC studies commonly follow three metric-selection approaches.
The first employs information-theoretic metrics for both functions, such as communication rate and sensing mutual information or sensing rate~\cite{kobayashi2018joint,ouyang2022integrated}. The second pairs a communication metric, most commonly the achievable rate, with a task-specific sensing metric: detection probability~\cite{hua2024integrated,dong2022sensing,li2026distributed}, the CRB for
estimation~\cite{liu2021optimaltx,dong2022sensing,liu2022joint,
hua2024crbrate,ren2024fundamental,xiong2023fundamental,
jointAngleDelay2023,li2024uplink,song2026crb}, tracking
error~\cite{wang2024resource,dong2022sensing,li2026distributed}, or ambiguity-related quantities such as delay--Doppler resolution and peak sidelobe level~\cite{li2025sensing}. The third uses signal-quality measures, most commonly communication SNR and sensing SNR~\cite{liu2022isacbeam,yu2023activeris,li2022novel,fan2024power}. Other works define ISAC-specific metrics such as value of service~\cite{li2024maximizing} and integration
gain~\cite{11161833}.

{To the best of our knowledge, current rate--CRB
optimizations do not explicitly account for the detection reliability
required for parameter estimates to become available}~
\cite{liu2022joint,ren2024fundamental,hua2024crbrate,
dong2022sensing,xiong2023fundamental,jointAngleDelay2023,
song2026crb,rsmaNF2025,liu2021optimaltx,li2024uplink}.
Detection and estimation are typically treated as separate sensing tasks, with a formulation selecting one metric and leaving the other unconstrained. Detection is typically interpreted as a one-shot decision establishing target presence in the search mode. Even works that address multiple sensing tasks in ISAC do not model the interdependence among them. The problems in
\cite{dong2022sensing} optimize detection probability, CRB, and tracking posterior CRB (PCRB) in separate task-specific formulations. The optimization in \cite{li2026distributed} detects unconnected vehicles and tracks a disjoint set of connected vehicles. Neither problem models how the detection outcome on a given frame affects the availability of parameter estimates and the tracking update for the same target.
\subsection{Contributions}
This gap raises the following questions: \emph{(i) When does a rate--CRB-optimal allocation also satisfy a given energy-detection probability level? (ii) When it does not, how does explicitly enforcing detection alter the rate-optimal resource allocation, and what operating regimes emerge? (iii) How does detection gating affect the accuracy of delivered estimates, downstream tracking, and achievable communication rate?}

To answer these questions, we study a joint power--bandwidth
allocation problem with the three metrics of interest: achievable rate, detection probability, and range CRB, under three formulations: (1) detection-blind (DB) rate--CRB, similar in nature to existing rate--CRB formulations in that it uses the nominal range CRB\footnote{Throughout this paper, the \emph{nominal} CRB refers to the standard CRB derived under the target-present hypothesis. The \emph{detection-gated} CRB, formally derived in
Section~\ref{sec:dte}, is the CRB conditioned on
successful energy detector output under the DTE receiver signal processing model.} alone as the sensing metric, but examined here from a \emph{detection-feasibility viewpoint};
(2) \emph{detection-aware (DA)} rate--CRB, a minimal augmentation of the DB formulation that enforces an energy-detection probability requirement; and (3) DTE rate--CRB for a \emph{DTE sensing receiver model with detection-gated CRB}. Thus, the DB--DA comparison isolates the resource-allocation effect of enforcing a detection requirement, whereas the DA--DTE comparison isolates the statistical effect of detector gating. The dependence of the range CRB on both power and bandwidth makes the comparison across these formulations, and the underlying mechanisms, particularly transparent, as detailed in Section~\ref{sec:db-da}. 

Our main contributions are summarized as follows.

\begin{itemize}
\item We identify a resource-composition ambiguity in DB
rate--CRB optimization: allocations with the same nominal range CRB can provide substantially different energy-detection probabilities. We derive the rate-optimal DB power--bandwidth allocation and establish a closed-form CRB threshold that determines when this operating point is also detection-feasible.
\item We formulate the DA rate--CRB problem by imposing a
prescribed detection probability while retaining the nominal CRB constraint. This formulation corresponds to assuming that both the detector and the estimator operate on every frame, without specifying how the receiver couples their outputs. The analysis identifies three operating regimes and characterizes how enforcing detection changes the rate-optimal power--bandwidth allocation and  the communication rate.

\item For a DTE receiver, building on the general
detection-conditioned lower-bound framework in \cite{chaumette2005influence} and its energy-detector specialization in \cite{chaumette2007conditional}, we derive a closed-form \emph{detection-gated range CRB} under our adopted ISAC signal model. We characterize its power--bandwidth consequences and validate the resulting selection factor through Monte Carlo simulations. We further show how detection probability and the accuracy of accepted measurements jointly determine the tracking PCRB. We further analyze its impact on resource allocation and achievable communication rate.

\item {We formulate multi-user, multi-target max--min
rate optimization problems with joint power and bandwidth allocation
under DB, DA, and DTE sensing constraints and establish their convexity.}
Numerical results show that CRB-only allocations can leave weaker targets
detection-infeasible, enforcing detection in DA incurs a communication-rate
cost, and {DTE can recover} some cost by requiring
fewer sensing resources to meet the conditional-CRB constraint on
detector-accepted frames.

\end{itemize}

The remainder of the paper is organized as follows. Section~II introduces the communication and sensing signal model. Section~III analyzes the DB and DA formulations. Section~IV derives the detection-gated CRB, studies tracking under intermittent measurements and presents the DTE rate--CRB optimization. Section~V formulates and evaluates multi-user, multi-target resource allocation, and Section~VI concludes the paper.

\section{System Model and Performance Metrics}
\label{sec:system}

Consider an ISAC base station (BS) equipped with $N_t$ transmit antennas and $N_r$ co-located sensing receive antennas, simultaneously serving $K$ communication users and sensing $Q$ targets. Communication
user $k$ has $N_{u,k}$ receive antennas. The communication user and sensing-target index sets are $\mathcal I_c=\{1,\ldots,K\}$ and $\mathcal I_s=\{1,\ldots,Q\}$, respectively. The BS transmits $M$ spatially precoded streams. Communication users and sensing targets
may be served by separate streams occupying disjoint subbands, or a user and target may be paired on the same stream. Both deployments are considered in Section~\ref{sec:resource-allocation}. In frame $\ell$,\footnote{From the sensing viewpoint, a frame defines one sensing-update interval and may comprise a coherent processing interval containing multiple pulses or OFDM symbols. The known transmitted symbols are included in the sensing reference for coherent processing~\cite{mirabella2023deterministic}} the transmitted signal is
\begin{equation}
\mathbf{x}_{\ell}(t)
=
\sum_{i=1}^{M}
\sqrt{p_{i,\ell}}\mathbf{f}_{i,\ell}
s_{i,\ell}(t;b_{i,\ell})
e^{j2\pi f^{c}_{i,\ell}t},
\qquad 0\leq t<T_f ,
\label{eq:tx-signal}
\end{equation}
where $T_f$ is the frame duration,
$\mathbf f_{i,\ell}\in\mathbb C^{N_t}$ is the unit-norm beamformer of stream $i$, and $p_{i,\ell}$ and $b_{i,\ell}$ are the power and bandwidth allocated to that stream, respectively. The stream allocations share the total resource budgets $P_t$ and $B_t$, such
that $\sum_{i=1}^{M}p_{i,\ell}\leq P_t$ and $\sum_{i=1}^{M}b_{i,\ell}\leq B_t$. The beamformers are held fixed during resource allocation. The unit-power baseband waveform $s_{i,\ell}(t;b_{i,\ell})$ has bandwidth $b_{i,\ell}$ and is centered at $f^c_{i,\ell}$. We assume a waveform with a rectangular baseband spectrum. For notational simplicity, the frame index is suppressed below except when temporal evolution is studied in Section \ref{sec:dte}.

\subsection{Communication Channel and Achievable Rate}

Let $\mathbf H_k\in\mathbb C^{N_{u,k}\times N_t}$ and $\mathbf u_k$ denote the channel and unit-norm receive combiner of communication user $k$, respectively. The channel is constant over one frame and the assigned subband. Its effective gain is assumed available for receive processing and resource allocation, while channel-estimation and synchronization errors are assumed to be negligible. With orthogonal subbands and noise power spectral density $N_{0,c}$, the noise power over bandwidth $b_k$ is $N_{0,c}b_k$. This is consistent with power--bandwidth allocation setups in~\cite{dong2022sensing,li2026distributed}. Defining $G_{c,k}=|\mathbf u_k^H\mathbf H_k\mathbf f_k|^2/N_{0,c}$, the post-combining communication SNR is $\gamma_{c,k}=G_{c,k}p_k/b_k$. Distinct streams occupy nonoverlapping, sufficiently confined subbands, as in \cite{dong2022sensing,li2026distributed}. The communication and sensing receivers isolate each subband before decoding or matched filtering, so cross-stream interference is neglected. All streams remain coupled through the total power and bandwidth budgets.

Therefore, the achievable rate on the subband assigned to user $k$ is given as follows:
\begin{equation}
R_k(p_k,b_k)
=
b_k\log_2\left(1+\frac{G_{c,k}p_k}{b_k}\right)
\quad [\mathrm{bps}].
\label{eq:comm-rate}
\end{equation}

\subsection{Monostatic Sensing Model and Performance Metrics}

The signal received at the co-located sensing array is represented as:
\begin{equation}
\mathbf y_s(t)
=
\sum_{q=1}^{Q}
\alpha_q
\mathbf a_r(\theta_q)
\mathbf a_t^H(\theta_q)
\mathbf x(t-\tau_q)e^{j2\pi\nu_qt}
+
\mathbf w_s(t),
\label{eq:sensing-observation}
\end{equation}
where $\alpha_q$, $\tau_q$, $\nu_q$, and $\theta_q$ are the complex reflectivity, delay, Doppler frequency, and angle of target $q$, respectively, with $\sigma_{\alpha,q}^2=\mathbb E[|\alpha_q|^2]$. The vectors $\mathbf a_t(\theta)$ and $\mathbf a_r(\theta)$ are the transmit and receive steering vectors, and $\mathbf w_s(t)$ contains receiver noise and clutter. The target parameters are assumed to remain approximately constant within each frame.

To isolate how sensing-performance characterization affects communication--sensing resource allocation, we assume that standard front-end clutter suppression and covariance-aware filtering have been performed~\cite[Secs.~5.2.3 and~9.4]{richards2014radar}. Residual clutter and thermal noise are modeled jointly as effective additive white Gaussian noise (AWGN). Clutter-aware ISAC receiver processing has also been studied~\cite{luo2024clutter,rivetti2025clutter}. The distinction among DB, DA, and DTE (Sections~\ref{sec:db-da} and~\ref{sec:dte}) remains applicable under such models because it concerns how detection and estimation enter the performance characterization, rather than the assumed clutter statistics. For signal-dependent clutter, the closed-form detection probability,
range CRB, and detection-gating/selection factor used here may not apply directly, and their model-specific re-derivation or numerical evaluation is beyond the scope of this work.

We focus on tracking-mode transmission following an initial search stage that has established a coarse delay--Doppler--angle cell for each target. In subsequent frames, the receiver validates the target echo and locally refines its parameters within the predicted cell. This paper studies resource allocation using the range CRB $\mathrm{CRB}_{d,q}$, while Doppler and angle are assumed available or estimated separately. Different targets are assumed to be separable through their spatial responses or delay--Doppler cells. Thus, initial search, association, and ambiguity resolution are outside the scope of this work. After spatial and subband separation, the samples retained for target $q$ are collected as:
\begin{equation}
\mathbf z_q
=
\sqrt{p_q}\alpha_q\mathbf s_q(\tau_q;b_q)
+
\mathbf v_q,
\label{eq:target-observation}
\end{equation}
where $\mathbf s_q(\tau_q;b_q)$ is the known sampled sensing waveform for target $q$ after fixed transmit and receive beamforming and subband separation. Its dependence on $\tau_q$ and $b_q$ determines the available range information. The post-processing noise-plus-clutter is
modeled as $\mathbf v_q\sim\mathcal{CN}(\mathbf 0,\sigma_{v,q}^2\mathbf I)$, where $\sigma_{v,q}^2$ accounts for thermal noise and residual clutter.

Following the Swerling-I model, the target reflectivity is constant within a frame and independent across frames, with
$\alpha_q\sim\mathcal{CN}(0,\sigma_{\alpha,q}^2)$. The waveform energy $\|\mathbf s_q(\tau_q;b_q)\|^2$ is held fixed as $b_q$ varies. Therefore, the post-processing echo and noise powers are $p_q\sigma_{\alpha,q}^2\|\mathbf s_q(\tau_q;b_q)\|^2$ and $\sigma_{v,q}^2$, respectively. The resulting sensing signal-to-clutter-plus-noise ratio (SCNR) is given as follows:
\begin{equation}
\gamma_{s,q}(p_q)
=
\underbrace{\frac{
\sigma_{\alpha,q}^2
\|\mathbf s_q(\tau_q;b_q)\|^2
}{
\sigma_{v,q}^2
}}_{\defn G_{s,q}} p_q
=
G_{s,q}p_q.
\label{eq:sensing-scnr}
\end{equation}
With unit-power waveform normalization and fixed $T_f$, varying $b_q$ changes the waveform's delay sensitivity without changing its energy. Hence, $G_{s,q}$ is independent of $b_q$, while bandwidth affects the range Fisher information.

\subsubsection{Framewise Detection}

Under nominal alignment with the target waveform, the normalized matched-filter output is $r_q =
{\mathbf s_q^H\mathbf z_q}/ ({\sigma_{v,q}\|\mathbf s_q\|}).$ The energy detector determines whether the current-frame observation contains a usable target echo by performing the following binary hypothesis test:
\begin{align}
\mathcal H_{0,q}:&\quad r_q=n_q,
&&\text{no usable target echo},\\
\mathcal H_{1,q}:&\quad r_q=\widetilde{\alpha}_q+n_q,
&&\text{target echo present},
\end{align}
where
$n_q\sim\mathcal{CN}(0,1)$ and $\widetilde{\alpha}_q
=\sqrt{p_q}\alpha_q\|\mathbf s_q\|/\sigma_{v,q}
\sim\mathcal{CN}(0,\gamma_{s,q})$. Averaging over the Swerling-I reflectivity gives $r_q\mid\mathcal H_{0,q}\sim\mathcal{CN}(0,1)$ and $r_q\mid\mathcal H_{1,q}\sim\mathcal{CN}(0,1+\gamma_{s,q})$. For these distributions, the Neyman--Pearson likelihood-ratio test is monotone in $|r_q|^2$ and is therefore equivalent to the energy test: $D_q=1$ if $|r_q|^2>\delta_q$, and $D_q=0$ otherwise, where $D_q=1$ denotes a positive detector decision. The corresponding detection event is $\mathcal D_q=\{D_q=1\}=\{|r_q|^2>\delta_q\}$. 

Since $|r_q|^2$ is exponentially distributed with mean one under $\mathcal H_{0,q}$, a prescribed false-alarm probability gives $\Pi_{FA,q}=e^{-\delta_q}$, or $\delta_q=-\ln\Pi_{FA,q}$. Under $\mathcal H_{1,q}$, its mean is $1+\gamma_{s,q}$, giving
\begin{equation}
\Pi_{D,q}(p_q)
=
\Pr\{\mathcal D_q\mid\mathcal H_{1,q}\}
=
\Pi_{FA,q}^{1/(1+G_{s,q}p_q)}.
\label{eq:detection-probability}
\end{equation}

\subsubsection{Nominal Range CRB}

For range estimation, let $\tau_q$ be a deterministic unknown parameter and $\alpha_q$ be a random nuisance parameter. Under the white Gaussian noise-plus-clutter model, the delay--reflectivity Fisher information cross term is proportional to $\int f|S_q(f)|^2\,df$ and therefore vanishes for a zero-centered baseband spectrum~\cite{vantrees2001detection}. For the assumed rectangular spectrum, the mean-squared bandwidth is $\beta_{\mathrm{rms},q}^2=b_q^2/12$. Using $d_q=c\tau_q/2$, the nominal range CRB is given as follows
\cite{liu2022joint}:
\begin{equation}
\mathrm{CRB}_{d,q}(p_q,b_q)
=
\frac{3c^2}
{8\pi^2G_{s,q}p_qb_q^2}
\quad [\mathrm{m}^2].
\label{eq:nominal-range-crb}
\end{equation}
The nominal CRB is not conditioned on the energy detector output. It assumes that every target-present frame, including weaker returns, is an estimation input once the initial search-mode target detection has succeeded. Thus, $\mathrm{CRB}_{d,q}$ and $\Pi_{D,q}$ characterize distinct aspects of framewise sensing performance whose resource-allocation consequences are studied next.

\section{Detection-Blind and Detection-Aware Rate--CRB  Analysis}
\label{sec:db-da}
This section studies two communication-rate maximization problems under different sensing-performance constraints. We first consider a DB rate--CRB formulation. Although the rate--CRB pair is widely used in ISAC, our objective is to examine it from a detection-feasibility perspective. Specifically, we determine whether a CRB-based allocation yields sufficient echo energy for the current-frame return to be reliably distinguished from noise and residual clutter. The analysis identifies the regimes in which the DB allocation can be insufficient. Motivated by this limitation, we formulate DA as an architecture-neutral characterization that imposes the nominal CRB and detection probability as parallel requirements, without specifying how the receiver couples their outputs. Their coupling in a DTE receiver is discussed later in Section \ref{sec:dte}.
\subsection{Single-Link Performance Metrics}
To keep the analysis tractable and expose the underlying power--bandwidth dependence mechanism, we consider a single communication user and a single sensing target served by a shared ISAC waveform. The multi-user, multi-target problem is studied later in Section~\ref{sec:resource-allocation}. The transmit and receive beamformers are fixed, with their array and propagation gains absorbed into the effective communication and sensing gains, $G_c$ and $G_s$. The waveform is assigned power $P$ and bandwidth
$B$, where $0<P\leq P_{\max}$ and $B_{\min}\leq B\leq B_{\max}$. For the signal model introduced in Section~\ref{sec:system}, the communication rate, detection probability, and range CRB for this single ISAC link are given as:
\begin{align}
R(P,B)
&=B\log_2\left(1+\frac{G_cP}{B}\right),
\label{eq:rate-single}\\
\Pi_D(P)
&=\Pi_{FA}^{1/(1+G_sP)},
\label{eq:pd-single}\\
\CRB_d(P,B)
&=\frac{3c^2}{8\pi^2G_sPB^2}.
\label{eq:crb-single}
\end{align}
For a given detection-probability level $\Pi_{FA}<\Pi_D^{\rm th}<1$, define $\gamma_s^{\rm th}\defn \ln\Pi_{FA}/\ln\Pi_D^{\rm th}-1$ and
$P_{\rm th}\defn\gamma_s^{\rm th}/G_s$. The same power and bandwidth determine all three quantities, but detection depends on power through $G_sP$, while the range CRB depends on $PB^2$. This difference is central to the analysis below. 
\subsection{DB Rate Optimization on a Constant-CRB Contour}
\label{sec:constant-crb}
We first consider communication-rate maximization for a prescribed range CRB. This formulation corresponds to a sensing model in which the target's presence has  been established within a range cell and \textit{the estimator processes every subsequent frame under the target-present model}, including frames containing weak target returns. We denote the CRB achieved by an allocation as $\CRB_d(P,B)$ and its prescribed value as $\overline{\CRB}_d$. Defining
$\kappa\defn 3c^2/(8\pi^2G_s\overline{\CRB}_d)$,
Equation~\eqref{eq:crb-single} shows that the $(P,B)$ allocations achieving this CRB form a one-parameter contour. With $\mathcal B=(0,P_{\max}]\times[B_{\min},B_{\max}]$, this contour is the following set:
\begin{equation}
\mathcal C(\kappa)
\defn
\left\{(P,B)\in\mathcal B:PB^2=\kappa\right\},
\label{eq:crb-contour}
\end{equation}
where $\mathcal C(\kappa)\neq \emptyset$ if $\kappa\leq P_{\max}B_{\max}^2$.
\begin{remark}[Purpose of the Constant-CRB Contour]
\label{rem:constant-crb}
Conventionally, a CRB requirement is written as an inequality $\CRB_d(P,B)\leq\overline{\CRB}_d$. In the single-user, single-target shared-resource setting considered here, maximizing rate subject only to this inequality drives power and bandwidth toward their available maxima, obscuring how different resource compositions providing the same CRB support the hierarchical sensing operation. Therefore, we use the constant-CRB contour as an analytical device to expose the effect of this resource-composition ambiguity. For the adopted sensing model, the allocations satisfying $PB^2=\kappa$ form the boundary of the CRB-feasible region $PB^2\geq\kappa$. Varying the prescribed CRB moves this boundary. In DA, the detection probability constraint restricts its feasible portion through $P\geq P_{\rm th}$, while detector gating in DTE changes its shape. Thus, the contours show how the feasible resource region and the rate-optimal allocation change across DB, DA, and DTE in the $(P,B)$ plane. Section~\ref{sec:resource-allocation} returns to conventional CRB inequality constraints and interprets the multi-user, multi-target allocation results relative to the corresponding constant-CRB contours.
\end{remark}

On the constant-CRB contour $\mathcal C(\kappa)$, the maximum achievable rate is given as:
\begin{equation}
R_{\rm DB}^{\star}(\kappa) \;\defn\;
\max_{(P,B)\in\mathcal C(\kappa)} R(P,B).
\label{eq:RDB-def}
\end{equation}
Since this optimization does not include a detection probability constraint, we call this the \textit{DB rate--CRB optimum}. Substituting $P = \kappa/B^2$
into~\eqref{eq:rate-single} gives
\begin{equation}
R(B;\kappa) = B\log_2\left(1 + \frac{G_c\kappa}{B^3}\right).
\label{eq:rate-contour}
\end{equation}
Setting $dR/dB = 0$ yields $\ln(1+u) = 3u/(1+u)$ with
$u \defn G_c\kappa/B^3$, whose unique positive root is
$u^\star \approx 15.80$. The interior optimum is realized at
\begin{equation}
\label{eq:db-interior}
B_{\rm DB}^{\star}(\kappa) = \Big(\frac{G_c\kappa}{u^\star}\Big)^{1/3},
\quad
P_{\rm DB}^{\star}(\kappa) = \frac{u^{\star\,2/3}\kappa^{1/3}}{G_c^{2/3}},
\end{equation}
when $(P_{\rm DB}^{\star}, B_{\rm DB}^{\star}) \in \mathcal B$.
The sensing SCNR at this optimum is given as:
\begin{equation}
\label{eq:db-snr}
    \gamma_{s,\rm DB}^{\star}(\kappa) = G_s\,P_{\rm DB}^{\star}(\kappa) =
G_s\,(u^{\star\,2}\,\kappa/G_c^2)^{1/3}.
\end{equation}
The optimization does not check whether detection is reliable at this point, so $R_{\rm DB}^{\star}(\kappa)$ is well-defined even when the detection probability is arbitrarily small. {\color{black} The complete box-constrained solution, including the minimum-bandwidth, maximum-power, and maximum-bandwidth boundary cases, is given in Appendix~\textbf{\ref{app:db-complete}}. Here, we focus on the interior solution that is relevant to the detection-feasibility mechanism.
}
\begin{figure}[t]
\centering
\includegraphics[width=\columnwidth]{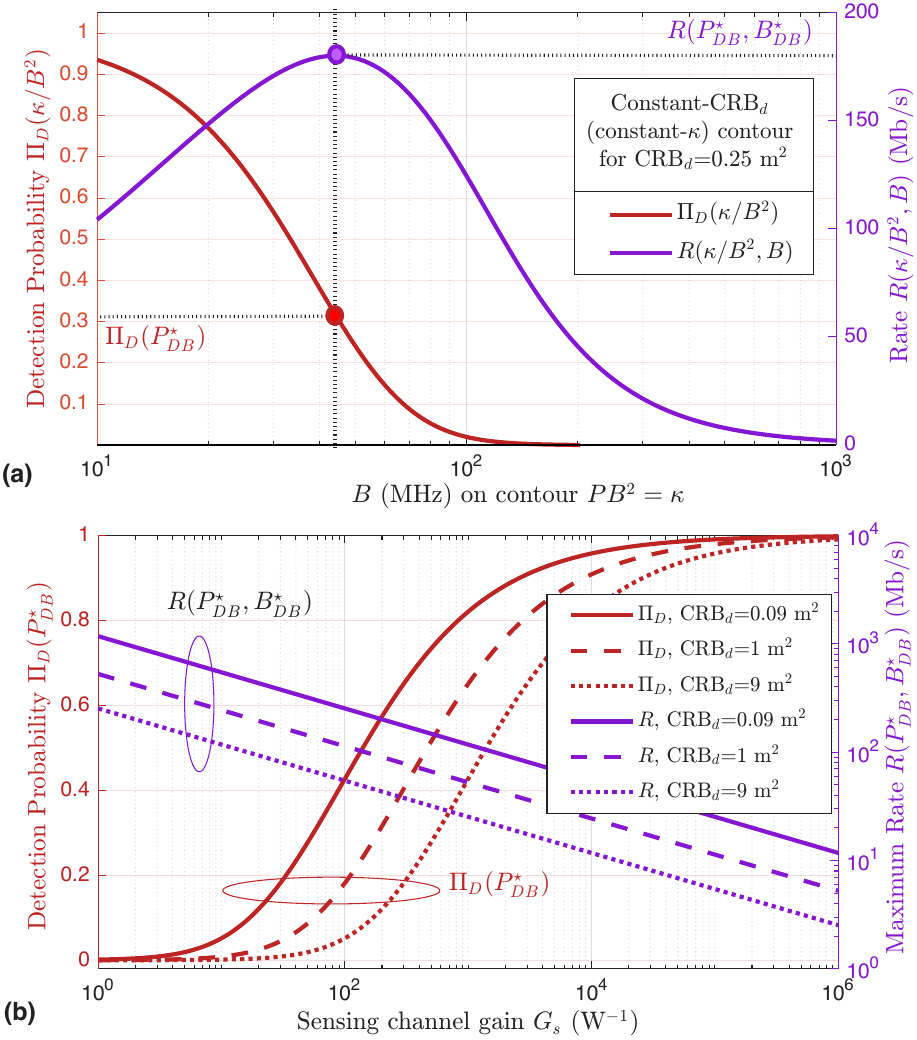}
\caption{\textbf{(a)} Rate and detection probability on a
constant-CRB contour; \textbf{(b)} rate and detection probability at the DB rate--CRB optimum versus sensing channel gain for different prescribed CRBs.}
\label{fig:db}
\end{figure}

Fig.~\ref{fig:db}(a) shows that on a 
constant-CRB contour, the DB rate--CRB optimum lands at an operating point with $\Pi_D$ far below the conventional radar design value $\Pi_D^{\rm th}=0.9$~\cite{richards2014radar}. The cause is structural:
a constant CRB is held by trading power for bandwidth along
$PB^2=\kappa$, so the rate maximizer can slide toward small $P$ and large $B$ without affecting CRB, while detection depends on {\color{black}the post-processing SCNR $\gamma_s=G_sP$ and therefore has no corresponding bandwidth lever.} Fig.~\ref{fig:db}(b) shows that this behavior persists across deployments: as $G_s$ shrinks, $\Pi_D$ at the DB rate--CRB optimum collapses sharply while $R_{\rm DB}^{\star}$ increases.

\subsubsection{Detection-Feasible CRB Constraint}
\label{sec:df-threshold}
Although detection does not enter the DB optimization, \eqref{eq:db-snr} gives its implied detection probability. We now determine whether this probability reaches a given level $\Pi_D^{\rm th}$ and identify the corresponding condition on the
prescribed range CRB.

\begin{proposition}[Detection-Feasible CRB Constraint]
\label{thm:critical-crb}
Assume that the interior solution in~\eqref{eq:db-interior} lies strictly inside $\mathcal B$. The DB optimum satisfies
$\Pi_D(P_{\rm DB}^{\star}(\kappa))\geq\Pi_D^{\rm th}$ if and only if
\begin{equation}
\begin{aligned}
\overline{\CRB}_d
\leq 
\CRB_{\rm DF}
\defn
\frac{3c^2u^{\star\,2}G_s^2}
{8\pi^2G_c^2(\gamma_s^{\rm th})^3}
=
\frac{3c^2u^{\star\,2}G_s^2}
{8\pi^2G_c^2
\left(
\displaystyle\frac{\ln\Pi_{FA}}
{\ln\Pi_D^{\rm th}}-1
\right)^3}.
\end{aligned}
\label{eq:critical-crb}
\end{equation}
Equivalently, the contour parameter must satisfy
$\kappa\geq\kappa_{\rm DF}$, where
$\kappa_{\rm DF}\defn G_c^2(\gamma_s^{\rm th})^3/
(u^{\star\,2}G_s^3)$.
\end{proposition}

\begin{proof}
$\Pi_D$ in~\eqref{eq:pd-single} is strictly increasing
in the sensing SCNR. Hence, $\Pi_D(P_{\rm DB}^{\star})\geq\Pi_D^{\rm th}$ is equivalent to $\gamma_{s,\rm DB}^{\star}\geq\gamma_s^{\rm th}$. Using
\eqref{eq:db-snr}, this condition becomes
$
G_s
\left(
{u^{\star\,2}\kappa}/{G_c^2}
\right)^{1/3}
\geq\gamma_s^{\rm th}.
$
Cubing both sides and rearranging gives
$\kappa\geq\kappa_{\rm DF}$. Substituting
$\kappa=3c^2/(8\pi^2G_s\overline{\CRB}_d)$ yields
\eqref{eq:critical-crb}.
\end{proof}

{\color{black}
The closed-form threshold in Proposition~\ref{thm:critical-crb} applies to the single-link, constant-CRB interior optimum considered here. The underlying relationship is more general: tightening a CRB constraint can indirectly increase the sensing SCNR and thereby improve the detection probability, even when detection is absent from the optimization. The corresponding transition in a multi-target problem also depends on the shared resource budgets and target parameters, as demonstrated in Section~\ref{sec:resource-allocation}.
}

From~\eqref{eq:db-snr},
$\gamma_{s,\rm DB}^{\star}\propto\overline{\CRB}_d^{-1/3}$, so reducing the prescribed CRB by a factor of eight only doubles the sensing SCNR; accordingly, $\CRB_{\rm DF}\propto(G_s/G_c)^2 (\gamma_s^{\rm th})^{-3}$ and becomes looser as the sensing channel strengthens relative to the communication channel. This slow inverse-cubic relationship makes CRB tightening an inefficient way to obtain detection reliability. Moreover, a stricter CRB constraint also draws more bandwidth in addition to power. Thus detection probability should be constrained explicitly when it must be guaranteed independently of $\overline{\CRB}_d$, as discussed next.

\subsection{Detection-Aware Rate--CRB Characterization}
\label{sec:da}

The preceding analysis shows that the operating point corresponding to the DB rate--CRB optimum may not yield an adequate detection probability. We now present a DA rate--CRB optimization that adds the detection-probability constraint $\Pi_D\geq\Pi_D^{\rm th}$ to the rate-maximization problem. We assume that the prescribed CRB is achievable within the operating region, i.e.,
$\kappa\leq P_{\max}B_{\max}^2$.

Define $B_{\rm th}(\kappa)=\sqrt{\kappa/P_{\rm th}}$ and
$\mathcal A(\kappa)\defn\mathcal C(\kappa)\cap\{P\geq P_{\rm th}\}$. When $\mathcal A(\kappa)\neq \emptyset$, the detection-aware (DA) operating rate is given as:
\begin{equation} 
R_{\rm DA}^{\star}(\kappa,\Pi_D^{\rm th})
=
\max_{(P,B)\in\mathcal A(\kappa)}R(P,B).
\label{eq:da-rate}
\end{equation}
Thus, the DB and DA operating points provide an identical  range CRB, while the DA point additionally satisfies a required detection probability.  Practically, this means that once target presence is established in a given range cell during an initial search, the energy detector and estimator operate on every subsequent frame. Here, we do not specify how their outputs are coupled; Section~\ref{sec:dte} then considers their coupling for a DTE receiver.

It is possible that the prescribed-CRB contour contains no
detection-feasible point even though the CRB and detection requirements can be jointly satisfied within the resource bounds. This occurs when $\kappa<P_{\rm th}B_{\min}^2$. In this case, $P\geq P_{\rm th}$ and $B\geq B_{\min}$ imply $PB^2>\kappa$, so satisfying the detection requirement necessarily over-delivers the CRB. Let $\mathcal F_{\rm DA}(\kappa) \defn\{(P,B)\in\mathcal B: PB^2\geq\kappa,\;P\geq P_{\rm th}\}$. When $\mathcal A(\kappa)=\emptyset$ but
$\mathcal F_{\rm DA}(\kappa)\neq \emptyset$, we select the off-contour operating point lexicographically as follows:
\begin{equation}
\begin{aligned}
\kappa_{\rm DA}^{\star}(\kappa)
&=
\min_{(P,B)\in\mathcal F_{\rm DA}(\kappa)}PB^2,\\
R_{\rm DA}^{\star}
&=
\max_{\substack{(P,B)\in\mathcal F_{\rm DA}(\kappa)\\
PB^2=\kappa_{\rm DA}^{\star}(\kappa)}}
R(P,B).
\end{aligned}
\label{eq:DA-lexicographic}
\end{equation}
Here, $\kappa_{\rm DA}^{\star}(\kappa)=\kappa$ when the equality contour is feasible, whereas $\kappa_{\rm DA}^{\star}(\kappa)>\kappa$ in the over-delivery regime. The first stage limits CRB over-delivery, and the second maximizes rate among the allocations attaining that minimum. For $\kappa<P_{\rm th}B_{\min}^2$, this gives
$\kappa_{\rm DA}^{\star}=P_{\rm th}B_{\min}^2$ and
$(P_{\rm DA}^{\star},B_{\rm DA}^{\star}) =(P_{\rm th},B_{\min})$. The achieved CRB is then $\CRB_d(P_{\rm th},B_{\min}) =3c^2/(8\pi^2G_sP_{\rm th}B_{\min}^2)
<\overline{\CRB}_d$. If $\mathcal F_{\rm DA}(\kappa)=\emptyset$, the two sensing requirements cannot be jointly satisfied within the available resources.
The complete DA feasibility characterization is provided in
Appendix~\textbf{\ref{app:da-complete}}. 

Combining the DB rate optimum with the minimum power $P_{\rm th}$ imposed by the detection-probability constraint yields the DA rate--CRB characterization. Depending on $\kappa$, the DA solution retains the DB optimum, moves along the constant-CRB contour to $(P_{\rm th},B_{\rm th})$, or leaves the contour and over-delivers the prescribed CRB at $(P_{\rm th},B_{\min})$. Under the interior-DB regime of Proposition~\ref{thm:critical-crb}, and assuming $P_{\rm th}\leq P_{\max}$, combining the DB optimum with the detection-imposed power floor gives the following optimal achievable rate:
\begin{equation}
R_{\rm DA}^{\star}
=
\begin{cases}
R_{\rm DB}^{\star}(\kappa),
&\kappa\geq\kappa_{\rm DF},\\[1mm]
R(P_{\rm th},B_{\rm th}),
&B_{\min}^2P_{\rm th}\leq\kappa<\kappa_{\rm DF},\\[1mm]
R(P_{\rm th},B_{\min}),
&\kappa<B_{\min}^2P_{\rm th}.
\end{cases}
\label{eq:da-regimes}
\end{equation}
The first row is the \emph{detection-feasible regime}, where the DB optimum already satisfies the detection constraint. In the \emph{contour-sliding regime}, the optimum moves along the prescribed-CRB contour to $(P_{\rm th},B_{\rm th})$. In the \emph{over-delivery regime}, the contour cannot reach $P_{\rm th}$ while satisfying $B\geq B_{\min}$, so the lexicographic extension selects $(P_{\rm th},B_{\min})$ and provides a CRB below the prescribed value. Since $\kappa_{\rm DF}\propto G_s^{-3}$, a weaker sensing channel reduces the DB-feasible region and expands the other two regimes.
\subsubsection{Rate--Detection Relationship}
\label{sec:rate-detection}

The complete rate--detection relation on $\mathcal C(\kappa)$ can be parameterized directly by the achieved detection probability. For $\Pi_D\in(\Pi_{FA},1)$, let $\gamma_s(\Pi_D)=\ln\Pi_{FA}/\ln\Pi_D-1$. Then $P=\gamma_s(\Pi_D)/G_s$ and, from $PB^2=\kappa$, $B=\sqrt{\kappa G_s/\gamma_s(\Pi_D)}$. Substitution into
\eqref{eq:rate-single} gives
\begin{equation}
R(\Pi_D;\kappa)
=
\sqrt{\frac{\kappa G_s}{\gamma_s(\Pi_D)}}
\log_2\left(
1+
\frac{G_c[\gamma_s(\Pi_D)]^{3/2}}
{G_s^{3/2}\sqrt{\kappa}}
\right).
\label{eq:rate-detection-explicit}
\end{equation}

Equation~\eqref{eq:rate-detection-explicit} describes how the communication rate changes with the achieved detection probability while the nominal range CRB remains fixed. If the DB optimum already satisfies $\Pi_D^{\rm th}$, its allocation remains unchanged. Otherwise, when $(P_{\rm th},B_{\rm th})$ lies within the resource bounds, the DA operating rate on the prescribed-CRB contour is $R(\Pi_D^{\rm th};\kappa)$. If this point is outside the feasible portion of the contour, the off-contour extension described above is required.

\subsubsection{Signed Rate Gap and Operating Regimes}
To quantify the communication-rate effect of adding the detection constraint to the prescribed-CRB problem, define
$\Delta R^{\star}(\kappa,\Pi_D^{\rm th})
\defn R_{\rm DB}^{\star}(\kappa)
-R_{\rm DA}^{\star}(\kappa,\Pi_D^{\rm th})$.
Also, let $\kappa_{B_{\min}}\defn u^{\star}B_{\min}^3/G_c$.

\begin{proposition}[Signed Rate Gap]
\label{prop:signed-gap}
The signed rate gap is: (a) zero in the DB-feasible regime; (b) strictly positive in the contour-sliding regime; and (c) in the over-delivery regime, strictly negative for $\kappa<\kappa_{B_{\min}}$, where both operating points lie at $B=B_{\min}$ and the DA point uses the higher
power $P_{\rm th}>\kappa/B_{\min}^2$, and of either sign for
$\kappa_{B_{\min}}\leq\kappa<B_{\min}^2P_{\rm th}$, where the DB point is interior while the DA point is $(P_{\rm th},B_{\min})$. The gap changes sign within the latter interval when it is nonempty.
\end{proposition}

\begin{proof}
The proof is provided in Appendix~\textbf{\ref{app:rate-gap}}.
\end{proof}

\begin{figure}[t]
\centering
\includegraphics[width=\columnwidth]{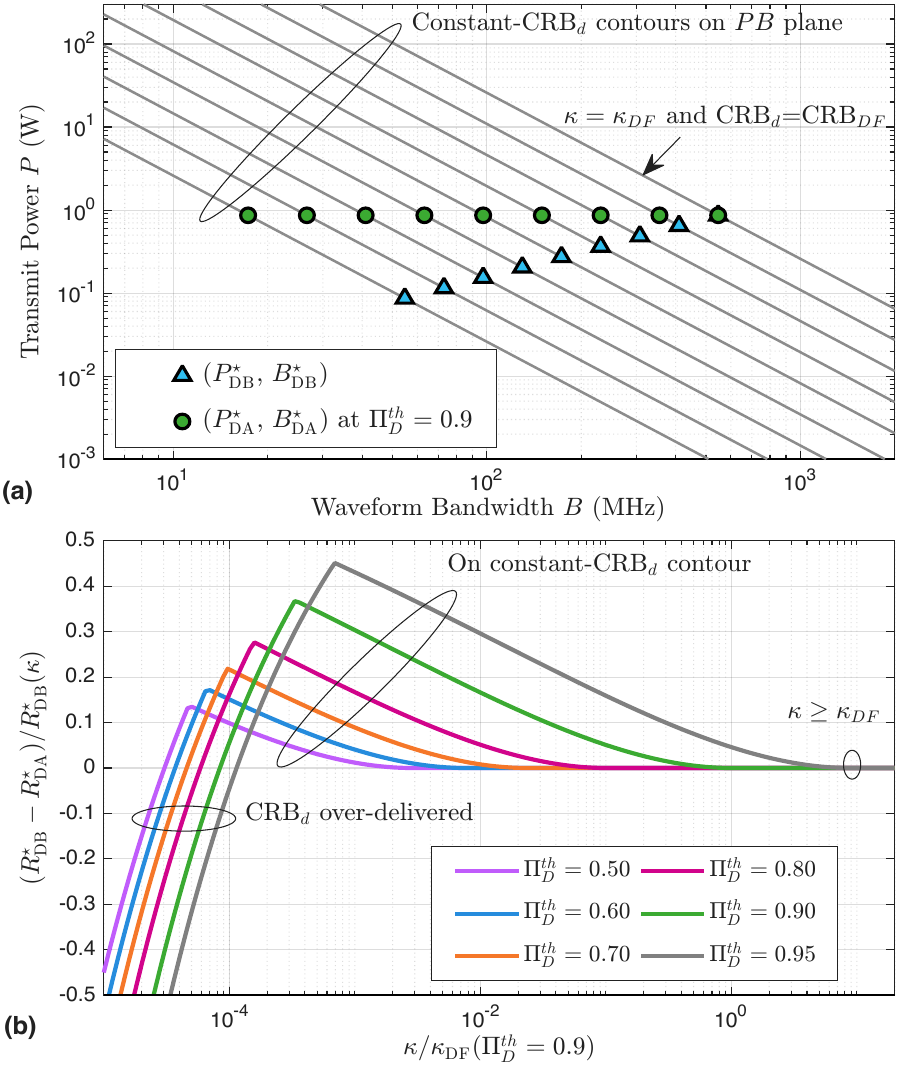}
\caption{Comparison of the DB and DA rate--CRB formulations:
\textbf{(a)} their rate-optimal $(P,B)$ allocations on constant-CRB contours; \textbf{(b)} the relative signed rate gap between their optima.}
\label{fig:db-da}
\end{figure}

Fig.~\ref{fig:db-da}(a) shows how the detection constraint changes the rate-optimal resource composition, while Fig.~\ref{fig:db-da}(b) shows the corresponding normalized gap $\Delta R^{\star}/R_{\rm DB}^{\star}(\kappa)$ for different values of $\Pi_D^{\rm th}$. In the DB-feasible regime, the two allocations coincide and the gap is zero. In the contour-sliding regime, DA moves toward higher power and lower bandwidth on the same constant-CRB
contour, producing a positive gap because this movement departs from the DB rate optimum. In the over-delivery regime, the contour cannot satisfy
the detection power floor with $B\geq B_{\min}$, so DA selects $(P_{\rm th},B_{\min})$ and achieves
$P_{\rm th}B_{\min}^2>\kappa$. The resulting lower CRB and additional power can make $R_{\rm DA}^{\star}>R_{\rm DB}^{\star}$, producing a negative gap. Note that this negative value reflects a comparison at unequal sensing accuracy rather than a rate gain from enforcing detection at the same CRB.

Section~\ref{sec:dte} next specifies a DTE receiver in which the hard energy-detector decision gates the range estimation.

\section{ISAC Performance under Detection-Gated Estimation and Tracking in a DTE Receiver}
\label{sec:dte}
In this section, we consider a DTE receiver, i.e., a signal-processing choice in which only the frames passing a threshold set by an energy detector are passed to the estimator \cite{chaumette2005influence, chaumette2007conditional}. On missed-detection frames, there is no estimate available. Using the framework in \cite{chaumette2005influence,chaumette2007conditional}, we derive a
closed-form expression for the detection-gated range CRB, denoted by CRB$_d^{\mathcal D}$, under the adopted ISAC signal and resource model
of Section~\ref{sec:system}. We then study how the feasible resource boundary changes relative to the DB and DA cases in Section \ref{sec:db-da}, and how optimal rate may be affected. Finally, we analyze the joint effect of detection probability and estimation error variance on tracking performance, and how the DB rate--CRB characterization can be insufficient from a tracker's viewpoint.

\subsection{Derivation of Detection-Gated Range CRB}
\label{sec:conditional-information}
We consider the post-target-acquisition tracking mode, where an initial search stage has established a coarse angle--delay--Doppler cell for the target. At each subsequent frame, the tracker predicts a range $d_0$ that
centers the matched-filter cell, and the detector validates the current-frame return in that cell. The true range is $d=d_0+\epsilon$, with $|2B\epsilon/c|\ll1$. Conditional on detection, the estimator refines the unknown within-cell displacement $\epsilon$ \cite{hua2024crbrate}.

Following the conditional-Fisher-information framework in
\cite{chaumette2005influence,chaumette2007conditional}, let
$\mathcal D\defn\{D=1\}$ denote the detector-acceptance event. The echo-bearing likelihood over accepted frames is given as follows:
\begin{equation}
p_{\mathcal D}(\mathbf z,\alpha\mid\boldsymbol\theta)
=
\frac{
p_1(\mathbf z\mid\boldsymbol\theta,\alpha)
p_{\alpha}(\alpha)
\mathbf 1_{\mathcal D}(\mathbf z)}
{\Pi_D(\boldsymbol\theta)},
\label{eq:conditional-likelihood}
\end{equation}
where $\boldsymbol\theta$ is the parameter(s) being estimated. The corresponding conditional Fisher information for range becomes:
\begin{equation}
\begin{aligned}
J_d^{\mathcal D}
&=
\mathbb E\left[
\left(
\frac{\partial}{\partial d}
\ln\!\left[
p_1(\mathbf z\mid d,\alpha)p_\alpha(\alpha)
\right]
\right)^2
\;\middle|\;
\mathcal D,\mathcal H_1
\right]\\
&
-
\left(
\frac{\partial}{\partial d}
\ln\Pi_D(d)
\right)^2.
\label{eq:conditional-fim-general}
\end{aligned}
\end{equation}
Thus, $(J_d^{\mathcal D})^{-1}$ bounds the MSE of an estimator that is locally unbiased under the distribution of the accepted observations.

For the post-clutter-suppression, white noise-plus-clutter model in Section \ref{sec:system},
\begin{equation}
\overline{\mathbf z}
=
\widetilde{\alpha}\mathbf u(d;B)+\mathbf n,
\qquad
\widetilde{\alpha}\sim\mathcal{CN}(0,\gamma_s),
\quad
\mathbf n\sim\mathcal{CN}(\mathbf0,\mathbf I),
\label{eq:conditional-whitened-model}
\end{equation}
where $\overline{\mathbf z}$ is the normalized observation and $\mathbf u(d;B)$ is its unit-norm range-dependent component, and $\gamma_s=G_sP$. The detector uses the template centered at $d_0$, giving $r=\widetilde{\alpha}
\varphi(d,d_0;B)+n_0$, where $\varphi(d,d_0;B)
=\mathbf u^H(d_0;B)\mathbf u(d;B)$ and $n_0\sim\mathcal{CN}(0,1)$. The acceptance event is $\mathcal D=\{|r|^2>\delta\}$. Since $r\mid\mathcal H_0\sim\mathcal{CN}(0,1)$, the statistic $|r|^2$ has a unit-mean exponential distribution, giving $\Pi_{FA}=\Pr\{|r|^2>\delta\mid\mathcal H_0\}=e^{-\delta}$ and hence $\delta=-\ln\Pi_{FA}$.

For the centered rectangular spectrum, $|\varphi(d_0+\epsilon,d_0;B)|^2 =1-\mathcal O((2B\epsilon/c)^2)$. Hence, small offsets from the predicted cell center change both the detection probability and the conditional range information only to second order in $\epsilon$. We therefore evaluate the local bound at $d=d_0$; this specifies the local evaluation point and does not imply that the receiver knows the current range. At this point, $r=\widetilde{\alpha}+n_0$. Since $\widetilde{\alpha}\sim\mathcal{CN}(0,\gamma_s)$ and
$n_0\sim\mathcal{CN}(0,1)$ are independent, $r\mid\mathcal H_1\sim\mathcal{CN}(0,1+\gamma_s)$, so $|r|^2$ is exponential with mean $1+\gamma_s$ and
$\Pi_D=\Pr\{|r|^2>\delta\mid\mathcal H_1\} =\exp[-\delta/(1+\gamma_s)]$.

Let $\dot{\mathbf u}=\partial\mathbf u/\partial d$ at $d_0$.
Without detector conditioning, the mean target energy is
$\mathbb E[|\widetilde{\alpha}|^2]=\gamma_s$, giving the nominal range information $J_d^{\rm nom}=2\gamma_s\|\dot{\mathbf u}\|^2$. For the centered spectrum, $\mathbf u^H\dot{\mathbf u}=0$, so $\Pi_D(d)$ is stationary at $d_0$ and $\partial\Pi_D/\partial d=0$.

Consequently, the second term in \eqref{eq:conditional-fim-general} vanishes at $d=d_0$. Therefore, detection-gating changes the range information through the energy distribution of the accepted frames. Thus,
\begin{equation}
J_d^{\mathcal D}
=
2\mathbb E[
|\widetilde{\alpha}|^2\mid\mathcal D]
\|\dot{\mathbf u}\|^2,
\label{eq:conditioned-range-information-proof}
\end{equation}
where $\|\dot{\mathbf u}\|^2=4\pi^2B^2/(3c^2)$ for the centered rectangular spectrum.

Now, for the detector-accepted ensemble,
\begin{equation}
\mathbb E[
|\widetilde{\alpha}|^2\mid\mathcal D]
=
\gamma_s
\left(
1+\frac{\delta\gamma_s}{(1+\gamma_s)^2}
\right),
\label{eq:conditional-target-energy}
\end{equation}
where the derivation of the conditional target energy over
detector-accepted frames is provided in Appendix~\textbf{\ref{app:conditional-energy}}.
Defining $h(\gamma_s)\defn1+\delta\gamma_s/(1+\gamma_s)^2$,
it follows that
\begin{equation}
J_d^{\mathcal D}
=
h(\gamma_s)J_d^{\rm nom},
\qquad
J_d^{\rm nom}
=
\frac{8\pi^2G_sPB^2}{3c^2},
\label{eq:conditional-range-information}
\end{equation}
Consequently,
\begin{equation}
\CRB_d^{\mathcal D}(P,B)
=
\frac{3c^2}
{8\pi^2G_sPB^2h(G_sP)}.
\label{eq:conditional-crb-explicit}
\end{equation}
Thus, detector gating scales the information in accepted frames by the \textit{selection factor} $h(\gamma_s)$ while retaining $PB^2$ as the underlying range-information factor. The factor satisfies $\lim_{\gamma_s\rightarrow0}h(\gamma_s)
=\lim_{\gamma_s\rightarrow\infty}h(\gamma_s)=1$ and reaches its maximum $h(1)=1+\delta/4$ when the target and effective clutter-plus-noise powers at the matched-filter output are equal. As $\gamma_s\rightarrow0$, $\Pi_D\rightarrow\Pi_{FA}$ and both the nominal and conditional information vanish. As $\gamma_s\rightarrow\infty$, $\Pi_D\rightarrow1$, almost every frame is accepted, and the conditional and nominal information coincide. This conditional information increase does not imply improved unconditional or per-transmitted-frame accuracy: missed frames deliver no range measurement, and their effect is represented separately through $\Pi_D$ and the intermittent-measurement tracking model below.
\begin{figure}[t]
\centering
\includegraphics[width=\columnwidth]{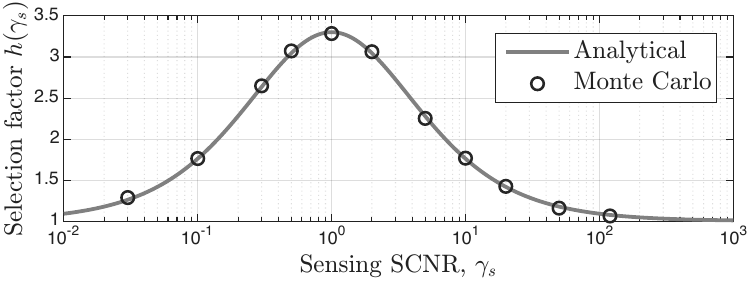}
\caption{Monte Carlo validation of the selection factor $h(\gamma_s)$.}
\label{fig:selection-factor}
\end{figure}

\subsubsection{Monte Carlo Validation of the Selection Factor}
For each SCNR $\gamma_s$, we generate $N_{\rm MC}$ independent Swerling-I coefficients
$\widetilde{\alpha}_m\sim\mathcal{CN}(0,\gamma_s)$ and white noise $n_m\sim\mathcal{CN}(0,1)$, and form
$r_m=\widetilde{\alpha}_m+n_m$. The detector accepts frame $m$ when $D_m=\mathbf 1\{|r_m|^2>\delta\}=1$, where
$\delta=-\ln\Pi_{FA}$. The number of accepted frames is $N_{\mathcal D}=\sum_{m=1}^{N_{\rm MC}}D_m$. The empirical selection factor is therefore
\begin{equation}
\widehat h(\gamma_s)
=
\frac{1}{\gamma_s N_{\mathcal D}}
\sum_{m=1}^{N_{\rm MC}}
|\widetilde{\alpha}_m|^2D_m.
\label{eq:empirical-selection-factor}
\end{equation}
Fig.~\ref{fig:selection-factor} compares the empirical estimate with $h(\gamma_s)=1+\delta\gamma_s/(1+\gamma_s)^2$, where $\delta=-\ln\Pi_{FA}$ is set using $\Pi_{FA}=10^{-4}$. The close agreement confirms the predicted target-energy enrichment in detector-accepted frames and the resulting increase in detection-gated range Fisher information.

\subsection{DTE Rate--CRB Performance Implications and Resource-Feasibility Consequences}
\label{sec:conditional-rate-implication}

Section~\textbf{\ref{sec:db-da}} used a constant nominal-CRB contour to characterize the DB and DA operating points. We now apply the same resource-composition perspective to DTE, with estimation accuracy measured by CRB$_d^\mathcal{D}$.

For a prescribed conditional CRB $\overline{\CRB}_d^{\mathcal D}$, define $\kappa_{\mathcal D}=3c^2/(8\pi^2G_s\overline{\CRB}_d^{\mathcal D})$.
For this single-target setting, let $\delta=-\ln\Pi_{FA}$ and $h(\gamma)=1+\delta\gamma/(1+\gamma)^2$.
Since $\CRB_d^{\mathcal D}\propto
[Ph(G_sP)B^2]^{-1}$, the DTE constant-CRB contour is
\begin{align}
\mathcal C_{\mathcal D}(\kappa_{\mathcal D})
=
\left\{
(P,B)\in\mathcal B:
Ph(G_sP)B^2=\kappa_{\mathcal D}
\right\},
\label{eq:dte-constant-crb-contour}
\end{align}
Let $\zeta(P)=Ph(G_sP)$. Since $\zeta'(P)=1+2\delta G_sP/(1+G_sP)^3>0$, $\zeta(P)$ is strictly increasing and $B_{\mathcal D}(P)=\sqrt{{\kappa_{\mathcal D}}/{Ph(G_sP)}}$ is strictly decreasing.

Because the detection-conditioned CRB characterizes the accuracy of accepted measurements but not their availability, DTE also requires $P\geq P_{\rm th}$. Define the detection-feasible portion of the DTE constant-CRB contour as $\mathcal A_{\mathcal D}(\kappa_{\mathcal D})
\defn
\mathcal C_{\mathcal D}(\kappa_{\mathcal D})
\cap\{P\geq P_{\rm th}\}.$
Following the least-over-delivery convention of Section~\ref{sec:da}, define $\mathcal F_{\mathcal D}(\kappa_{\mathcal D})
\defn\{(P,B)\in\mathcal B:
Ph(G_sP)B^2\geq\kappa_{\mathcal D},\,
P\geq P_{\rm th}\}$ and select
\begin{equation}
\begin{aligned}
\kappa_{\rm DTE}^{\star}(\kappa_{\mathcal D})
&=
\min_{(P,B)\in\mathcal F_{\mathcal D}(\kappa_{\mathcal D})}
Ph(G_sP)B^2,\\
R_{\rm DTE}^{\star}
&=
\max_{\substack{
(P,B)\in\mathcal F_{\mathcal D}(\kappa_{\mathcal D})\\
Ph(G_sP)B^2=
\kappa_{\rm DTE}^{\star}(\kappa_{\mathcal D})}}
R(P,B).
\end{aligned}
\label{eq:dte-rate-optimization}
\end{equation}
Here, $\kappa_{\rm DTE}^{\star}(\kappa_{\mathcal D}) =\kappa_{\mathcal D}$ when the equality contour is feasible, whereas $\kappa_{\rm DTE}^{\star}(\kappa_{\mathcal D}) >\kappa_{\mathcal D}$ in the over-delivery regime. In the latter case, the CRB-controlling product is $Ph(G_sP)B^2$, rather than $PB^2$ as in DA. The least-over-delivery convention of Section~\ref{sec:da} selects $(P_{\rm DTE}^{\star},B_{\rm DTE}^{\star}) =(P_{\rm th},B_{\min})$. The dependence of $h(G_sP)$ on power prevents a closed-form solution. Since $B=B_{\mathcal D}(P)$ is uniquely determined by each power value on $\mathcal C_{\mathcal D}(\kappa_{\mathcal D})$, \eqref{eq:dte-rate-optimization} is a bounded one-dimensional optimization over the feasible power values and is solved numerically.

\begin{figure}[t]
\centering
\includegraphics[width=\columnwidth]{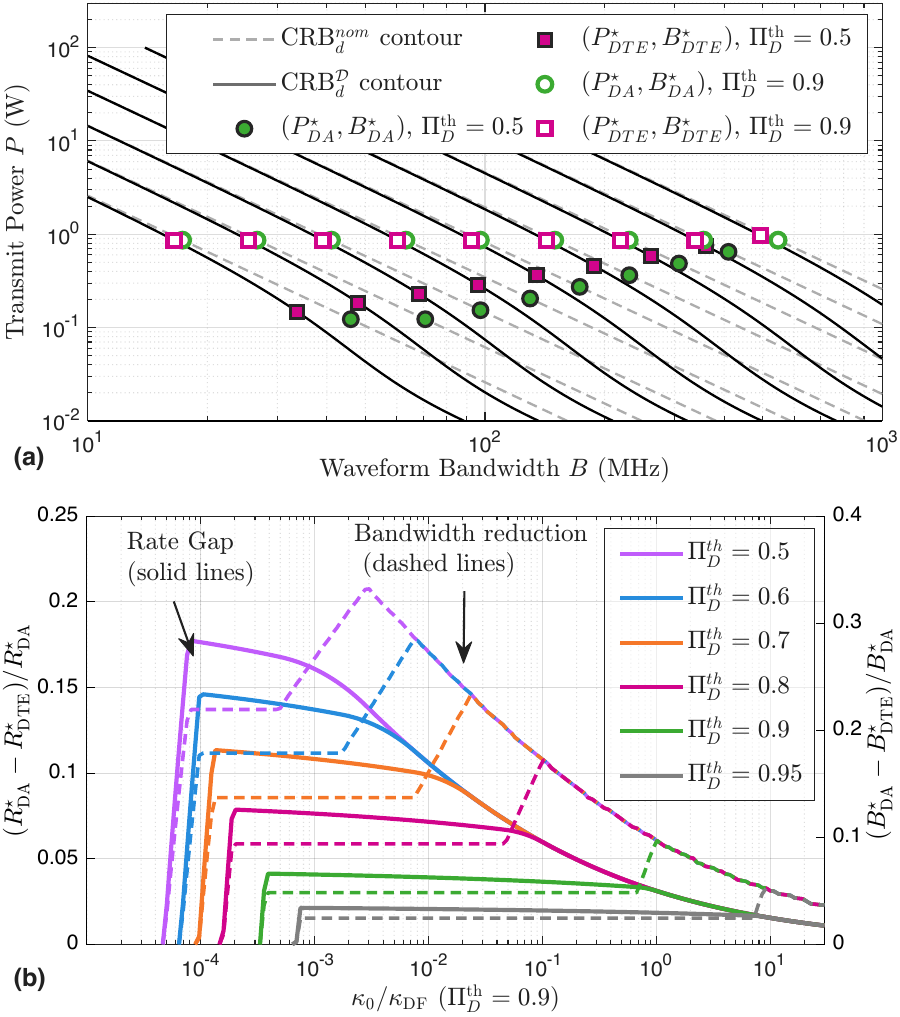}
\caption{Comparison of DA and DTE rate--CRB characterizations:
\textbf{(a)} nominal and detection-conditioned constant-CRB contours with their detection-constrained rate-maximizing operating points;
\textbf{(b)} relative rate difference (left y-axis) and relative bandwidth reduction (right y-axis) for the same numerical accuracy target, applied to detector-accepted measurements in DTE, versus prescribed CRB, for different detection probability requirements.}
\label{fig:dte-rate-comparison}
\end{figure}

To compare DA and DTE, we impose the same numerical accuracy target, interpreted as $\CRB_d^{\rm nom}$ in DA and $\CRB_d^{\mathcal D}$ in
DTE, and denote the associated  parameter by
$\kappa_0$. It defines $PB^2=\kappa_0$ for DA and
$Ph(G_sP)B^2=\kappa_0$ for DTE, giving
$B_{\mathcal D}(P)=B_{\rm DA}(P)/\sqrt{h(G_sP)}
\leq B_{\rm DA}(P)$. Fig.~\ref{fig:dte-rate-comparison} shows zero rate difference and bandwidth reduction for
$\kappa_0\leq\kappa_1\defn P_{\rm th}B_{\min}^2$. Both differences increase for $\kappa_1<\kappa_0<\kappa_2\defn
P_{\rm th}h(G_sP_{\rm th})B_{\min}^2$, as DA reaches its equality contour before DTE, and decrease beyond their peaks as the required SCNR increases and $h(G_sP)\rightarrow1$. Thus, DTE can require less bandwidth to attain the prescribed accuracy on detector-accepted estimates, without implying better all-frame sensing accuracy, since missed frames provide no estimate and availability remains governed by $\Pi_D$.

\subsection{Tracking with Intermittent Measurements}
\label{sec:tracking-pd}
The CRB$_d^{\mathcal D}$ bounds the variance of delivered range measurements. A downstream tracker also depends on which frames deliver those measurements. Existing ISAC tracking formulations evaluate the PCRB under the implicit assumption that a measurement is available in every frame~\cite{dong2022sensing,wang2024resource,yuan2021bayesian}.

For $\Pi_D<1$, let $\ell$ denote the frame index and consider a constant-velocity model with state
$\boldsymbol\xi_\ell=[d_\ell,\dot d_\ell]^T$, where $d_\ell$ and $\dot d_\ell$ are the range and velocity, respectively:
\begin{equation}
\boldsymbol\xi_\ell
=
\mathbf F\boldsymbol\xi_{\ell-1}
+\boldsymbol\omega_{\ell-1},
\quad
\mathbf F=
\begin{bmatrix}
1&T_f\\
0&1
\end{bmatrix},
\quad
\mathbf G=
\begin{bmatrix}
T_f^2/2\\
T_f
\end{bmatrix}.
\label{eq:tracking-model}
\end{equation}
Here, $\boldsymbol\omega_\ell\sim\mathcal N(\mathbf0,\mathbf Q)$ is the acceleration-induced process noise, with $\mathbf Q=\sigma_a^2\mathbf G\mathbf G^T$ and $\sigma_a^2$ denoting the acceleration variance.

For a controlled tracking benchmark, each detector-accepted range estimate is modeled as locally unbiased with Gaussian error variance equal to the prescribed CRB$_d^{\mathcal D}$. Thus, a detected frame provides
$y_\ell=\mathbf H\boldsymbol\xi_\ell+v_\ell$, where
$\mathbf H=[1\;0]$ and $v_\ell\sim\mathcal N(0,R_d)$, with
$R_d=\overline{\CRB}_d^{\mathcal D}$.

Assuming independent reflection and noise realizations across frames, $D_\ell\sim\operatorname{Bernoulli}(\Pi_D)$ denotes measurement availability. Following the intermittent-observation Kalman-filter model in~\cite{sinopoli2004kalman}, let $\mathbf P_\ell$ and $\mathbf S_\ell$ denote the posterior and predicted covariances, respectively. Their random Riccati recursion is given as:
\begin{align}
\mathbf S_\ell
&=
\mathbf F\mathbf P_{\ell-1}\mathbf F^T+\mathbf Q,
\notag\\
\mathbf P_\ell
&=
\mathbf S_\ell
-
D_\ell\mathbf S_\ell\mathbf H^T
\left(
\mathbf H\mathbf S_\ell\mathbf H^T+R_d
\right)^{-1}
\mathbf H\mathbf S_\ell .
\label{eq:random-riccati}
\end{align}
The Riccati recursion is the frame-by-frame update of the state-estimation error covariance, alternating between uncertainty growth during prediction and uncertainty reduction when a measurement is available.
For this linear-Gaussian model, $\mathbf P_\ell$ equals the conditional PCRB for the realized measurement-availability sequence \cite{schenato2007foundations}.
We evaluate the sequence-averaged steady-state PCRB:
\begin{equation}
\overline{\mathbf P}_\infty(\Pi_D)
=
\lim_{\ell\rightarrow\infty}
\mathbb E[\mathbf P_\ell].
\label{eq:tracking-steady-pcrb}
\end{equation}
One delivered measurement contributes range information $1/R_d$;
therefore, its expected per-frame contribution is $\Pi_D/R_d$
\cite{niu2001pcrlb}.
The same Bernoulli process produces prediction-only gaps. If $N$ denotes the number of consecutive missed frames before the next detected frame, then $\Pr\{N=n\}=(1-\Pi_D)^n\Pi_D$. During a gap of $n$ missed
frames,
\begin{equation}
\mathbf P_{\ell+n|\ell}
=
\mathbf F^n\mathbf P_{\ell|\ell}(\mathbf F^n)^T
+
\sum_{i=0}^{n-1}
\mathbf F^i\mathbf Q(\mathbf F^i)^T.
\label{eq:prediction-gap-covariance}
\end{equation}
For the constant-velocity model, its range component is
\begin{align}
[\mathbf P_{\ell+n|\ell}]_{1,1}
={}&
[\mathbf P_{\ell|\ell}]_{1,1}
+
2nT_f[\mathbf P_{\ell|\ell}]_{1,2}
+
n^2T_f^2[\mathbf P_{\ell|\ell}]_{2,2}
\notag\\
&+
\frac{\sigma_a^2T_f^4}{12}n(4n^2-1).
\label{eq:cv-gap-range}
\end{align}
The derivation of~\eqref{eq:cv-gap-range} is given in
Appendix~\textbf{\ref{app:cv-gap-derivation}}. Thus, decreasing $\Pi_D$ both reduces the average measurement-information arrival rate and increases the probability of long prediction-only gaps
\cite{schenato2007foundations}.
\begin{figure}[t]
\centering
\includegraphics[width=\columnwidth]{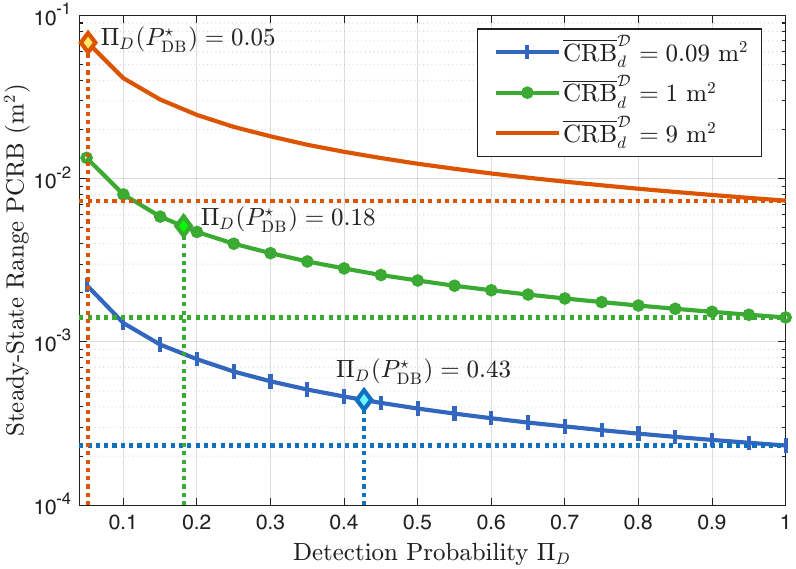}
\caption{Steady-state range PCRB of a constant-velocity Kalman tracker with acceleration uncertainty $\sigma_a=1$~m/s$^2$ and Bernoulli-intermittent measurements at frame duration $T_f=1$~ms.}
\label{fig:tracking-pd}
\end{figure}
Fig.~\ref{fig:tracking-pd} shows the steady-state PCRB obtained by Monte Carlo evaluation of~\eqref{eq:random-riccati}. A lower prescribed $\overline{\CRB}_d^{\mathcal D}$ reduces the tracking error at every $\Pi_D$, while decreasing $\Pi_D$ increases the PCRB. The DB markers show the tracking PCRB at the detection probabilities induced by CRB-only allocations; their separation from the $\Pi_D=1$ endpoint quantifies the degradation hidden by assuming that every frame provides a measurement. Thus, for a DTE receiver, tracking performance is jointly determined by the conditional CRB, which specifies the accuracy of accepted measurements, and $\Pi_D$, which determines the frequency and spacing of measurement updates.
\subsection{Relative Trends of the DB, DA, and DTE Formulations and Insights on Generality}
\label{sec:DA-approximates-DTE}
Under limited resource budgets, the system may operate with a lower detection-probability requirement. In this case, $h(\gamma_s)$ can differ appreciably from one. For a gated receiver, DTE can require less bandwidth than DA to meet the same numerical accuracy target on detector-accepted measurements. As $\bar\Pi_D^{\rm req}$ approaches one, the required SCNR increases and $h(\gamma_{s})\to1$, causing the DA and DTE optima to converge. A stricter CRB constraint similarly drives the DB rate-optimal allocation toward higher sensing SCNR, making the detection constraint less
consequential. Thus, \emph{relaxing the CRB constraint enlarges the DB--DA gap}, while \emph {lowering the $\Pi_D^{\rm th}$ requirement makes the DA--DTE gap more pronounced}.

For the common estimation accuracy constraint $\kappa_0$, let $\mathcal F_X(\kappa_0)$ denote the feasible set of resource pairs $(P,B)$ under formulation $X\in\{\mathrm{DB},\mathrm{DA},\mathrm{DTE}\}$.
Since $B_{\mathcal D}(P)=B_{\rm DA}(P)/\sqrt{h(G_sP)}
\leq B_{\rm DA}(P)$, $\mathcal F_{\rm DA}\subseteq\mathcal F_{\rm DB}$ and $\mathcal F_{\rm DA}\subseteq\mathcal F_{\rm DTE}$. Therefore, DA is the most conservative formulation in terms of resource feasibility.

The closed-form threshold in Section \ref{sec:db-da} and the DTE selection factor derived above follow from the adopted rectangular-spectrum, Swerling-I, and white post-suppression noise-plus-clutter models. The underlying conclusion is more general: whenever detection and estimation depend differently on the allocated resources, satisfying an estimation-only metric requirement like the CRB does not ensure reliable detection.
\section{Multi-Target and Multi-User Resource Allocation}
\label{sec:resource-allocation}
Sections~\ref{sec:db-da} and~\ref{sec:dte} analyzed the three formulations, DB, DA, and DTE, in a single-user, single-target setting. We now extend the analysis to a multi-user, multi-target deployment in which one ISAC BS simultaneously serves multiple communication users and sensing targets under joint power and bandwidth budgets, as described in Section \ref{sec:system}. For DB, DA, and DTE, we maximize the minimum per-user achievable rate subject to the corresponding sensing constraints in shared and orthogonal resource deployments. We establish the convexity of the resulting problems, describe their solution methods, and present the numerical results.

\subsection{Problem Formulation}

We adopt a max--min achievable rate objective to ensure fairness among users, subject to sensing service constraints corresponding to DB, DA, and DTE formulations, under shared power and bandwidth budgets. Consider the user and target index sets, resource variables, and performance metrics defined in Section~\ref{sec:system}. For the sensing constraints, define:
\begin{equation}
J_{d,q}^{\rm nom} \defn \frac{1}{\CRB_{d,q}^{\rm nom}}, \quad
J_{d,q}^{\mathcal D} \defn \frac{1}{\CRB_{d,q}^{\mathcal D}}, \quad
\overline J_q \defn \frac{1}{\overline{\CRB}_{d,q}}.
\end{equation}
The three formulations then share the following common form:

\begin{equation}
(\mathcal P^{(v)}):\quad
\begin{aligned}
\max_{\mathbf p,\mathbf b}\quad
&\min_{u\in\mathcal I_c}\; R_u(p_u,b_u)\\
\mathrm{s.t.}\quad
&\mathrm C_q^{(v)},\quad q\in\mathcal I_s,\\
&\mathbf 1^\top\mathbf p\le P_t,\quad
 \mathbf 1^\top\mathbf b\le B_t,\quad
 \mathbf p,\mathbf b\succeq\mathbf 0,
\end{aligned}
\label{eq:ra-rate-native}
\end{equation}
where $v\in\{\mathrm{DB},\mathrm{DA},\mathrm{DTE}\}$. The sensing constraints for each variant are:
\begin{equation}
\begin{array}{ccl}
\mathrm C_q^{(\mathrm{DB})}
&:&J_{d,q}^{\rm nom}(p_q,b_q)\ge \overline J_q,\\[1mm]
\mathrm C_q^{(\mathrm{DA})}
&:&J_{d,q}^{\rm nom}(p_q,b_q)\ge \overline J_q,\ \ 
   \Pi_{D,q}(p_q)\ge \Pi_{D,q}^{\rm th},\\[1mm]
\mathrm C_q^{(\mathrm{DTE})}
&:&J_{d,q}^{\mathcal D}(p_q,b_q)\ge \overline J_q,\ \ 
   \Pi_{D,q}(p_q)\ge \Pi_{D,q}^{\rm th}.
\end{array}
\label{eq:ra-rate-service-variants}
\end{equation}
The DB variant here serves as a reference formulation and follows the
CRB-only sensing characterization commonly used in ISAC resource
allocation, e.g.,~\cite{dong2022sensing}.

\subsection{Convex Reformulation and Solution}
\label{sec:ra-convexity}

Under the closed forms of $J_{d,q}^{\rm nom}$, $J_{d,q}^{\mathcal D}$, and $\Pi_{D,q}$, normalize $\overline p_u = p_u/P_t$, $\overline b_u = b_u/B_t$ for each user and $\overline p_q = p_q/P_t$, $\overline b_q = b_q/B_t$ for each target, and define
\begin{equation}
\rho_u = \frac{G_{c,u} P_t}{B_t}, \quad
g_q = G_{s,q} P_t, \quad
\lambda_q = \frac{\kappa_q P_t B_t^{2}}{J_0}, \quad
\eta_q = \ln\frac{\overline J_q}{J_0},
\label{eq:ra-log-constants}
\end{equation}
where $J_0$ is a fixed reference with units of range Fisher information and $\delta_q = -\ln \Pi_{{\rm FA},q} > 0$. In the normalized variables, the achievable rate becomes:
\begin{equation}
R_u(\overline p_u, \overline b_u)
= B_t\, \overline b_u \log_2\!\left(1 + \frac{\rho_u \overline p_u}{\overline b_u}\right),
\label{eq:ra-normalized-rate}
\end{equation}
and, since $J_{d,q}^{\rm nom}$ is a monomial in $(p_q, b_q)$, the log-domain sensing information takes the following form:
\begin{align}
f_q^{\rm nom}(\overline p_q, \overline b_q)
&\defn \ln \frac{J_{d,q}^{\rm nom}}{J_0}
= \ln \lambda_q + \ln \overline p_q + 2\ln \overline b_q,
\label{eq:ra-log-nominal}\\
f_q^{\mathcal D}(\overline p_q, \overline b_q)
&\defn f_q^{\rm nom} + \ln h_q(g_q \overline p_q).
\label{eq:ra-log-detection-conditioned}
\end{align}
The detection requirement in DA and DTE is equivalent to the power floor:
\begin{equation}
\overline p_q \ge \overline p_{D,q}^{\rm th} \defn
\frac{1}{g_q}\!\left(\frac{\ln \Pi_{{\rm FA},q}}{\ln \Pi_{D,q}^{\rm th}} - 1\right), \quad
\Pi_{{\rm FA},q} < \Pi_{D,q}^{\rm th} < 1.
\label{eq:ra-detection-power-floor}
\end{equation}
Introducing the variable $t$ yields the following epigraph form:
\begin{equation}
(\widetilde{\mathrm P}_{\rm R}^{(v)}):\quad
\begin{aligned}
\max_{\overline{\mathbf p},\overline{\mathbf b},\,t}\quad & t\\
\mathrm{s.t.}\quad
& R_u(\overline p_u, \overline b_u)\ge t, && u\in\mathcal I_c,\\
& \widetilde{\mathrm C}_q^{(v)}, && q\in\mathcal I_s,\\
& \mathbf 1^\top\overline{\mathbf p}\le 1,\ \
  \mathbf 1^\top\overline{\mathbf b}\le 1,\ \
  \overline{\mathbf p},\overline{\mathbf b}\succeq\mathbf 0,
\end{aligned}
\label{eq:ra-rate-log}
\end{equation}
with log-domain sensing constraints:
\begin{equation}
\begin{array}{ccl}
\widetilde{\mathrm C}_q^{(\mathrm{DB})}
&:& f_q^{\rm nom}\ge\eta_q,\\[1mm]
\widetilde{\mathrm C}_q^{(\mathrm{DA})}
&:& f_q^{\rm nom}\ge\eta_q,\ \ \overline p_q\ge\overline p_{D,q}^{\rm th},\\[1mm]
\widetilde{\mathrm C}_q^{(\mathrm{DTE})}
&:& f_q^{\mathcal D}\ge\eta_q,\ \ \overline p_q\ge\overline p_{D,q}^{\rm th}.
\end{array}
\label{eq:ra-log-service-variants}
\end{equation}

\subsubsection{Convexity and Solver}
The rate $R_u$ is jointly concave in $(\overline p_u,\overline b_u)$ because its Hessian is negative semidefinite. The nominal sensing information $f_q^{\rm nom}$ is the sum of $\ln\overline p_q$ and $2\ln\overline b_q$ and is therefore concave. Concavity of $f_q^{\mathcal D}$ reduces to that of $\psi_{\delta_q}(x)\defn\ln[x\,h_q(x)]$ at $x=g_q\overline p_q$; differentiation gives
\begin{equation}
\psi_{\delta_q}''(x) = -\frac{N_{\delta_q}(x)}{x^2(1+x)^2\left[(1+x)^2+\delta_q x\right]^2},
\label{eq:psi-second}
\end{equation}
with $N_{\delta_q}(x)=1+(6+2\delta_q)x+(15+12\delta_q+2\delta_q^2)x^2+(20+18\delta_q+4\delta_q^2)x^3+(15+8\delta_q)x^4+6x^5+x^6>0$ for $x,\delta_q>0$, so $\psi_{\delta_q}''(x)<0$ and $f_q^{\mathcal D}$ is jointly concave. The remaining constraints are affine, so $(\widetilde{\mathrm P}_{\rm R}^{(v)})$ is convex for all three variants.

DB and DA are solved using CVX in MATLAB~\cite{cvx}. Although the DTE problem is also convex, CVX does not recognize the factor $h_q(g_q\overline p_q)$ as a disciplined-convex-programming (DCP)-compliant expression. Therefore, we solve the DTE variant using the sequential quadratic programming (SQP) algorithm implemented by MATLAB \texttt{fmincon}. Since every local optimum of a convex problem is globally optimal, successful convergence yields the global DTE solution up to numerical tolerance ~\cite{mathworksFmincon,boyd2004convex}. In every simulated instance, \texttt{fmincon} terminated with both the first-order optimality measure and the maximum constraint violation below $10^{-9}$.

\subsection{Simulation Setup}
\label{sec:V-sim-setup}

We normalize the range-information requirements by the reference value $J_0 \defn \frac{8\pi^2G_{s,\mathrm{ref}}P_{\mathrm{ref}}B_{\mathrm{ref}}^2} {3c^2}= \frac{8\pi^2}{27}~\mathrm{m}^{-2} \approx 2.92~\mathrm{m}^{-2},$ obtained using $G_{s,\mathrm{ref}}=1~\mathrm{W}^{-1}$, $P_{\mathrm{ref}}=1~\mathrm{W}$, and $B_{\mathrm{ref}}=100~\mathrm{MHz}$. This value corresponds to a range accuracy of $1/\sqrt{J_0}\approx0.6~\mathrm{m}$.

We consider $M_Q=5$ targets with heterogeneous sensing gains\footnote{Using the radar equation $G_{s,q}=G_tG_r\lambda_c^{2}\sigma_{\mathrm{RCS},q}T_f/[(4\pi)^{3}d_q^{4}N_0L_s]$, with $f_c=30$~GHz, $G_t=G_r=20$~dBi, $\sigma_{\mathrm{RCS},q}=1~\mathrm{m}^{2}$,
$T_f=1$~ms, $N_0=-162$~dBm/Hz, and $L_s=6$~dB, the selected sensing gains correspond to target ranges of approximately $\{450,530,670,710,800\}~\mathrm{m}$ ~\cite{dong2022sensing,richards2014radar}.} 
$\mathbf G_s=[50,25,10,8,6]~\mathrm{W}^{-1}$ and normalized range Fisher information requirements $\overline{\mathbf J}=[1.00,0.50,0.15,0.10,0.06]$, where $\overline J_q$ specifies the required Fisher information as a multiple of $J_0$. Accordingly, target~$T_1$ has the strictest range-accuracy requirement, $\sqrt{\overline{\CRB}_{d,1}}\approx0.6~\mathrm{m}$, while target~$T_5$ has the least stringent requirement, $\sqrt{\overline{\CRB}_{d,5}}\approx2.4~\mathrm{m}$. The parameters are selected to expose different allocation regimes within the same deployment. Stronger, closer targets are assigned tighter range-information requirements, reflecting a safety-oriented application. The detector parameters are $\Pi_{{\rm FA},q}=10^{-5}$ and, unless otherwise stated, $\Pi_{D,q}^{\rm th}=0.5$ for every target. We consider two deployments:

$\bullet$ \emph{Shared deployment:} Five target-directed streams are transmitted. The streams for $T_1$, $T_2$, and $T_3$ also carry data for $M_U=3$ communication users located along the same beam directions, although possibly at different ranges. The streams for $T_4$ and $T_5$ are dedicated to sensing. We use $(P_t,B_t) =(10~\mathrm{W},70~\mathrm{MHz})$. 

$\bullet$ \emph{Orthogonal deployment:}  The $M_U=5$ communication streams and five sensing streams occupy disjoint subbands under the common power and bandwidth budgets. Here, we use $(P_t,B_t) =(20~\mathrm{W},200~\mathrm{MHz})$.

We use $\mathbf G_c=[345,244,154,138,119]~\mathrm{MHz/W}$, with the shared case using only the first three entries. These values correspond to a one-way communication link budget, whereas the sensing gains $\mathbf G_s$ account for two-way propagation.

The shared budget is deliberately tight so that the sensing constraints materially affect the allocations while all five targets remain serviceable. The larger orthogonal budget accommodates the ten disjoint streams. For Fig.~\ref{fig:rategap}(b), we additionally sweep
$P_t\in\{10,15,20\}~\mathrm{W}$ at $B_t=200~\mathrm{MHz}$ in the orthogonal deployment.
\begin{figure*}[t]
\centering
\includegraphics[width=\textwidth]{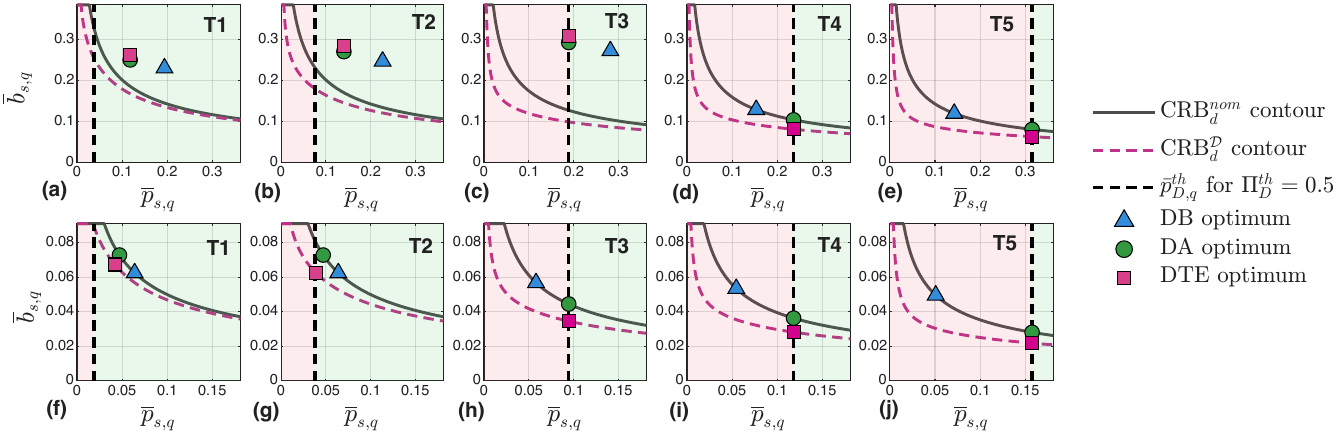}
\caption{Per-target sensing allocations for the heterogeneous five-target scenario. The top row, panels (a)--(e), shows the five-stream shared setting, in which $T_1$--$T_3$ streams also serve communication users, and the bottom row, panels (f)--(j), shows the allocations for the ten-stream orthogonal setting.}
\label{fig:allocations}
\label{fig:V1}
\end{figure*}

\begin{figure}
\centering
\includegraphics[width=\columnwidth]{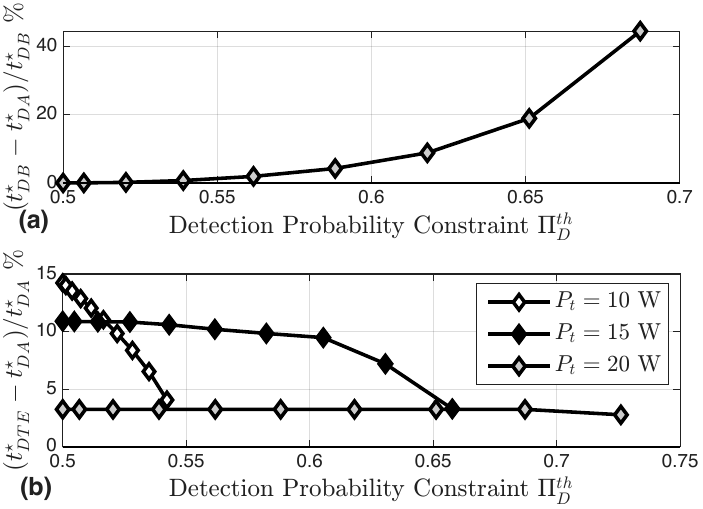}
\caption{Relative differences between the rate optima in the orthogonal
communication--sensing deployment versus the detection-probability
requirement: (a) DA versus DB; (b) DTE versus DA.}
\label{fig:rategap}
\end{figure}
\subsection{Numerical Results and Discussion}
Figs. \ref{fig:allocations} and \ref{fig:rategap} show the results for the DB, DA, and DTE rate--CRB optimization problems using the parameters specified above.
\subsubsection{Allocations and Operating Regimes}
Fig.~\ref{fig:allocations} shows the per-target sensing allocations in the $(P,B)$ planes. In each panel, the CRB$_d^{\rm nom}$ contour $\kappa_q\overline p_q\overline b_q^{2}=\overline J_q$ and the CRB$_d^{\mathcal D}$ contour $h_q(g_q\overline p_q)\kappa_q\overline p_q\overline b_q^{2}=\overline J_q$ are drawn together with the per-target detection floor $\overline p_{D,q}^{\rm th}$. The three variants place their optima at substantially different points.

\textit{Shared case:} Figs.~\ref{fig:allocations}(a)--(e) show the five targets split by strength. The three stronger targets $T_1,T_2,T_3$ are paired with the $M_U=3$ communication users; because $\overline p_{D,q}^{\rm th}$ is small for these targets, the CRB constraint alone admits detection-feasible optima. The rate objective assigns the paired streams additional power and bandwidth, placing their optima inside the sensing-feasible region rather than on the CRB boundary. In contrast, $T_4$ and $T_5$ use dedicated sensing streams, for which the CRB constraint binds and their allocations are realized along the corresponding CRB contours.

\textit{Orthogonal case:} Figs.~\ref{fig:allocations}(f)--(j) show that all five targets fall in the contour-sliding regime, which is a consequence of the orthogonal sensing-communication resource split. Therefore, the optimizer uses the smallest resource footprint that keeps each target both CRB- and detection-feasible. Consequently, no target is placed in an over-delivery regime, and $\overline p_q=\overline p_{D,q}^{\rm th}$ holds for every $q$.


\textit{Implication for rate--CRB formulations:} Across both settings, the DB rate--CRB optimum is detection-feasible for stronger targets ($T_1$, $T_2$, $T_3$) with tighter CRB requirements but infeasible for weaker targets ($T_4$, $T_5$) with looser requirements, confirming that a CRB constraint alone does not guarantee detection. The same structural limitation applies to rate--CRB formulations that omit an explicit detection constraint~\cite{liu2021optimaltx,dong2022sensing,liu2022joint,hua2024crbrate,ren2024fundamental,xiong2023fundamental,jointAngleDelay2023,li2024uplink,song2026crb}: their operating points may support detection under favorable conditions, but the formulation itself provides no such guarantee.

\subsubsection{Gap between Rate Optima: DA--DB and DTE--DA}
\textit{Shared case:} In the shared-resource simulations, communication streams are paired with the stronger, stricter-CRB targets, while weaker targets use dedicated sensing streams and exhibit contour sliding; hence, the shared-stream rate gap is negligible, consistent with the detection-feasible regime identified in Section~\ref{sec:db-da}. 

\textit{Orthogonal case:} Fig. \ref{fig:rategap} shows the rate gap results for the orthogonal communication-sensing resources. Each curve terminates at the maximum $\Pi_D^{\rm th}$ that can be supported for all targets under the corresponding power budget. Fig. \ref{fig:rategap}(a) shows the rate cost of enforcing detection in addition to CRB, i.e.,$(t^\star_{\rm DB} - t^\star_{\rm DA})/t^\star_{\rm DB}$ versus $\bar\Pi_D^{\rm th}$. The gap grows monotonically with $\bar\Pi_D^{\rm th}$ because a higher detection probability requirement forces the sensing streams to draw more power from the shared budget. Fig.~\ref{fig:rategap}(b) shows the DTE margin $(t^\star_{\rm DTE}-t^\star_{\rm DA})/t^\star_{\rm DA}$ for different power budgets. Each curve decreases with $\bar\Pi_D^{\rm th}$ because the selection factor $h(\gamma_s)\rightarrow 1$ as the sensing SCNR increases. When $h(\gamma_s)$ remains appreciably above one, DTE can meet the prescribed accuracy of detector-accepted measurements using less sensing bandwidth than DA, leaving more resources for communication.The curves do not reach zero because the weakest-target streams remain close to the minimum SCNR required for detection. Their SCNR is therefore not high enough for $h(\gamma_{s,q})$ to approach one, so the bandwidth reduction for accepted frames persists. Thus, all three curves approach a rate margin of approximately $3\%$.

Overall, the substantial DB--DA gap reflects the resource and rate cost of enforcing detection, whereas DTE provides a receiver-specific refinement of the architecture-neutral DA formulation, with its greatest
benefit under scarce resources. Although Sections~\ref{sec:db-da} and~\ref{sec:dte} use constant-CRB contours for analytical insight, the results confirm that the DB--DA effect is not created by the equality constraint; it persists under the conventional inequality constraints. Competition for shared resources places allocations on or near the corresponding contours, while additional resources move them deeper into the interior, where detection is more likely to be feasible.

\section{Conclusion}
This paper examined how CRB-based ISAC resource allocation supports the connected sensing operations of detection, estimation, and tracking. The single-link analysis showed that identical nominal range CRBs can arise from power--bandwidth allocations with substantially different detection probabilities. We derived when the DB rate--CRB optimum is detection-feasible and characterized how enforcing detection changes the allocation and communication rate. For a DTE receiver, we derived the detection-gated CRB of accepted frames and showed that their accuracy and availability jointly determine tracking performance. The multi-user, multi-target results further demonstrated that detection awareness can materially change resource distribution. These results motivate treating sensing as a hierarchical inference problem in which each metric is interpreted through the receiver operations that make its output available and meaningful, rather than
describing sensing by a single metric.
\label{sec:conclusion}


\appendices

\section{Complete Box-Constrained DB Solution}
\label{app:db-complete}

On $\mathcal C(\kappa)$, substituting $P=\kappa/B^2$ reduces the DB problem to~\eqref{eq:rate-contour}. The feasible bandwidth interval
is $B_{\rm L}(\kappa)\leq B\leq B_{\max}$, where
$B_{\rm L}(\kappa)\defn
\max\{B_{\min},\sqrt{\kappa/P_{\max}}\}$. The contour is nonempty if
and only if $\kappa\leq P_{\max}B_{\max}^2$.

Let $B_0(\kappa)\defn(G_c\kappa/u^\star)^{1/3}$. The complete box-constrained DB solution is given as:
\begin{equation}
B_{\rm DB}^{\star}(\kappa)
=
\begin{cases}
B_{\rm L}(\kappa),
&B_0(\kappa)<B_{\rm L}(\kappa),\\[1mm]
B_0(\kappa),
&B_{\rm L}(\kappa)\leq B_0(\kappa)\leq B_{\max},\\[1mm]
B_{\max},
&B_0(\kappa)>B_{\max},
\end{cases}
\label{eq:db-complete}
\end{equation}
with $P_{\rm DB}^{\star}(\kappa)
=\kappa/[B_{\rm DB}^{\star}(\kappa)]^2$.

\section{Complete Detection-Aware Feasibility Characterization}
\label{app:da-complete}
Recall that $\mathcal A(\kappa)=\mathcal C(\kappa)\cap\{P\geq P_{\rm th}\}$. We have $\mathcal A(\kappa)\neq\emptyset$ if and only if $P_{\rm th}\leq P_{\max}$ and $P_{\rm th}B_{\min}^2 \leq \kappa \leq P_{\max}B_{\max}^2.$ Under the interior-DB assumption of Proposition~\ref{thm:critical-crb}, the complete DA operating point is
\begin{equation}
(P_{\rm DA}^{\star},B_{\rm DA}^{\star})
=
\begin{cases}
(P_{\rm DB}^{\star},B_{\rm DB}^{\star}),
&\kappa\geq\kappa_{\rm DF},\\[1mm]
(P_{\rm th},B_{\rm th}),
&P_{\rm th}B_{\min}^2\leq\kappa<\kappa_{\rm DF},\\[1mm]
(P_{\rm th},B_{\min}),
&\kappa<P_{\rm th}B_{\min}^2.
\end{cases}
\label{eq:da-complete}
\end{equation}
The last branch uses the least-over-delivery convention
in~\eqref{eq:DA-lexicographic}. Finally, $\mathcal F_{\rm DA}(\kappa)\neq\emptyset$ if $P_{\rm th}\leq P_{\max}$ and $\kappa\leq P_{\max}B_{\max}^2$.

\section{Gap between the DA and DB Rate Optima}
\label{app:rate-gap}
Part (a) follows from the first row of~\eqref{eq:da-regimes}.

For part (b), the DA point in the contour-sliding regime is
$(P_{\rm th},B_{\rm th})$, with $B_{\rm th}<B_{\rm DB}^{\star}$ on the same contour. Since $R(B;\kappa)$ is strictly increasing on
$[0,B_{\rm DB}^{\star}]$, $R(P_{\rm th},B_{\rm th})<R_{\rm DB}^{\star}$ and therefore $\Delta R^{\star}>0$.

For part (c), when $\kappa<\kappa_{B_{\min}}$, both operating points are pinned at $B=B_{\min}$. The DB and DA rates are $R(\kappa/B_{\min}^2,B_{\min})$ and $R(P_{\rm th},B_{\min})$, respectively. Since $P_{\rm th}>\kappa/B_{\min}^2$ and rate is strictly increasing in power at fixed bandwidth, $\Delta R^{\star}<0$.

For $\kappa_{B_{\min}}\leq\kappa<B_{\min}^2P_{\rm th}$, the DB point is interior while the DA point is $(P_{\rm th},B_{\min})$. The gap is continuous in $\kappa$ across the endpoints of this interval. At $\kappa=\kappa_{B_{\min}}$, the DB point is $(u^\star B_{\min}/G_c,B_{\min})$, which reduces to the previous case and gives $\Delta R^{\star}<0$. At $\kappa=B_{\min}^2P_{\rm th}$, $B_{\rm th}=B_{\min}$ and the DA point reduces to the contour-sliding solution, which gives $\Delta R^{\star}>0$. Therefore, by the intermediate value theorem, $\Delta R^{\star}$ vanishes at some $\kappa\in(\kappa_{B_{\min}},B_{\min}^2P_{\rm th})$.

\section{Conditional Target Energy on Detector-Accepted Frames}
\label{app:conditional-energy}
At the centered operating point $d=d_0$, the detector output is $r=\widetilde{\alpha}+n_0$, where $\widetilde{\alpha}\sim\mathcal{CN}(0,\gamma_s)$ and $n_0\sim\mathcal{CN}(0,1)$ are independent. Hence, $r\sim\mathcal{CN}(0,1+\gamma_s)$.

Since $\widetilde{\alpha}$ and $r$ are jointly complex Gaussian, the conditional distribution of $\widetilde{\alpha}$ given $r$ has mean
$\mathbb E[\widetilde{\alpha}\mid r] =\gamma_s r/(1+\gamma_s)$ and variance $\gamma_s/(1+\gamma_s)$. Therefore,
\begin{equation}
\mathbb E[
|\widetilde{\alpha}|^2\mid r]
=
\frac{\gamma_s}{1+\gamma_s}
+
\frac{\gamma_s^2}{(1+\gamma_s)^2}|r|^2.
\label{eq:app-conditional-energy-r}
\end{equation}
The detector-acceptance event is $\mathcal D=\{|r|^2>\delta\}$, where $\delta=-\ln\Pi_{FA}$. Since $|r|^2$ is exponentially distributed with mean $1+\gamma_s$, the memory-less property of the exponential distribution gives
$
\mathbb E[
|r|^2\mid\mathcal D]
=
\delta+1+\gamma_s.
$
Taking the conditional expectation of \eqref{eq:app-conditional-energy-r} over $\mathcal D$ yields
\begin{align}
\mathbb E[
|\widetilde{\alpha}|^2\mid\mathcal D]
&=
\frac{\gamma_s}{1+\gamma_s}
+
\frac{\gamma_s^2}
{(1+\gamma_s)^2}
(\delta+1+\gamma_s)
\notag\\
&=
\gamma_s
\underbrace{\left(
1+
\frac{\delta\gamma_s}
{(1+\gamma_s)^2}
\right)}_{\defn h(\gamma_s)}.
\label{eq:app-conditional-target-energy}
\end{align}
Substitution into the conditional range information in \eqref{eq:conditional-range-information} yields $J_d^{\mathcal D} =h(\gamma_s)J_d^{\rm nom}$. The factor $h(\gamma_s)$ reflects selection of higher-energy echo-bearing observations by the detector.

\section{Prediction-Gap Range Covariance}
\label{app:cv-gap-derivation}

For the constant-velocity model in~\eqref{eq:tracking-model}, $\mathbf F^n=[\begin{smallmatrix}1&nT_f\\0&1\end{smallmatrix}]$. Hence, the range component of the propagated posterior covariance is
\begin{align}
\left[\mathbf F^n\mathbf P_{\ell|\ell}
(\mathbf F^n)^T\right]_{1,1}
={}&[\mathbf P_{\ell|\ell}]_{1,1}
+2nT_f[\mathbf P_{\ell|\ell}]_{1,2}
\notag\\
&+n^2T_f^2[\mathbf P_{\ell|\ell}]_{2,2}.
\label{eq:appendix-propagated-covariance}
\end{align}
Since $\mathbf Q=\sigma_a^2\mathbf G\mathbf G^T$ and $\mathbf F^i\mathbf G=[T_f^2(i+1/2)\;\;T_f]^T$, the accumulated
range-process-noise variance is
\begin{align}
\left[\sum_{i=0}^{n-1}
\mathbf F^i\mathbf Q(\mathbf F^i)^T\right]_{1,1}
&=\sigma_a^2T_f^4\sum_{i=0}^{n-1}(i+1/2)^2
\notag\\
&=\frac{\sigma_a^2T_f^4}{12}n(4n^2-1).
\label{eq:appendix-accumulated-noise}
\end{align}
Combining the propagated posterior covariance and accumulated process-noise covariance in~\eqref{eq:prediction-gap-covariance}
gives~\eqref{eq:cv-gap-range}.

\bibliographystyle{IEEEtran}
\bibliography{references}

\end{document}